\documentclass[amsmath, amssymb,letterpaper,longbibliography,twocolumn,accepted=2025-04-01]{quantumarticle}
\pdfoutput=1
\usepackage{graphicx}
\usepackage{amsthm}
\usepackage{amsfonts}
\usepackage{amssymb}
\usepackage{xcolor}
\usepackage[colorlinks=true,citecolor=blue]{hyperref}
\usepackage[capitalise]{cleveref}
\usepackage{dsfont}
\usepackage{stmaryrd}
\usepackage[linesnumbered, ruled, vlined]{algorithm2e}

\def\dsl{\llbracket}
\def\dsr{\rrbracket}
\newcommand{\Lambdaone}{\Lambda^{\{1\}}}
\newcommand{\Lambdaonetwo}{\Lambda^{\{1,2\}}}

\newcommand{\sufrak}{{\mathfrak su}}
\newcommand{\Lambdaonehat}{\hat{\Lambda}^{\{1\}}}
\newcommand{\Lambdaonetwohat}{\hat{\Lambda}^{\{1,2\}}}

\newcommand{\Lcal}{\mathcal{L}}
\newcommand{\Lcalt}{\widetilde{\mathcal{L}}}
\newcommand{\Lcalthat}{\hat{\widetilde{\mathcal{L}}}}
\newcommand{\Lcalhat}{\hat{\mathcal{L}}}
\newcommand{\Dcal}{\mathcal{D}}
\newcommand{\Hq}{\hat{H}_{q}}
\newcommand{\Hint}{\hat{H}_{\mathrm{int}}}
\newcommand{\Hdr}{\hat{H}_{\mathrm{dr}}}

\newcommand{\omegaqj}{\omega_{q_j}}
\newcommand{\sigmap}{\hat{\sigma}_{+}}
\newcommand{\sigmam}{\hat{\sigma}_{-}}

\newcommand{\sigmazj}{\hat{\sigma}_{z,j}}
\newcommand{\sigmapj}{\hat{\sigma}_{+,j}}
\newcommand{\sigmamj}{\hat{\sigma}_{-,j}}

\newcommand{\sigmapk}{\hat{\sigma}_{+,k}}
\newcommand{\sigmamk}{\hat{\sigma}_{-,k}}

\newcommand{\Tphij}{T_{\varphi,j}}
\newcommand{\Tonej}{T_{1,j}}

\newcommand{\Omegadrj}{\Omega_{\mathrm{dr},j}}
\newcommand{\omegadrj}{\omega_{\mathrm{dr},j}}
\newcommand{\phidrj}{\phi_{\mathrm{dr},j}}

\newcommand{\Dinfinput}{D^{\mathrm{inf}}_{|\Psi_{\mathrm{in}}\rangle}}

\newcommand{\Psiin}{\Psi_{\mathrm{in}}}

\newcommand{\Nthat}{\hat{\mathcal{N}}}
\newcommand{\Vt}{\mathcal{V}}
\newcommand{\Nt}{\mathcal{N}}

\newcommand{\Napproxhatk}{\hat{\mathcal{N}}_{\mathrm{approx},k}}
\newcommand{\Napproxhattwo}{\hat{\mathcal{N}}_{\mathrm{approx},2}}
\newcommand{\Napproxhatthree}{\hat{\mathcal{N}}_{\mathrm{approx},3}}
\newcommand{\Napproxhat}{\hat{\mathcal{N}}_{\mathrm{approx}}}
\newcommand{\Napprox}{\mathcal{N}_{\mathrm{approx}}}
\newcommand{\Nactualhat}{\hat{\mathcal{N}}_{\mathrm{actual}}}
\newcommand{\Nactual}{\mathcal{N}_{\mathrm{actual}}}

\newcommand{\Vapprox}{\mathcal{V}_{\mathrm{approx}}}
\newcommand{\Vapproxhat}{\hat{\mathcal{V}}_{\mathrm{approx}}}

\newcommand{\Vactual}{\mathcal{V}_{\mathrm{actual}}}
\newcommand{\Vactualhat}{\hat{\mathcal{V}}_{\mathrm{actual}}}

\newcommand{\pe}{p_{\mathcal{E}}}
\newcommand{\pf}{p_{\mathcal{F}}}
\newcommand{\pef}{p_{\mathcal{E},\mathcal{F}}}
\newcommand{\Videal}{\mathcal{V}_{\mathrm{ideal}}}
\newcommand{\Videalhat}{\hat{\mathcal{V}}_{\mathrm{ideal}}}

\newcommand{\Ucalhat}{\hat{\mathcal{U}}}
\newcommand{\Ucald}
{\mathcal{U}^\dagger}

\newcommand{\Lcalonekt}{\hat{\widetilde{\mathcal{L}}}_k}

\newcommand{\Lcalonenonet}{\hat{\widetilde{\mathcal{L}}}^{1-\mathrm{body}}}
\newcommand{\Lcalonentwot}{\hat{\widetilde{\mathcal{L}}}^{2-\mathrm{body}}}

\newcommand{\Lcalonennt}{\hat{\widetilde{\mathcal{L}}}^{n-\mathrm{body}}}
\newcommand{\Lcalonenkt}{\hat{\widetilde{\mathcal{L}}}^{k-\mathrm{body}}}
\newcommand{\Lcalonenmt}{\hat{\widetilde{\mathcal{L}}}^{m-\mathrm{body}}}
\newcommand{\Lcalonent}{\hat{\widetilde{\mathcal{L}}}}
\newcommand{\Lcalonethreet}{\hat{\widetilde{\mathcal{L}}}}
\newcommand{\Hhat}{\hat{H}}
\newcommand{\Ajhat}{\hat{A}_j}
\newcommand{\Ajhatd}{\hat{A}_j^\dagger}
\newcommand{\Aonehat}{\hat{A}_1}
\newcommand{\Aonehatd}{\hat{A}_1^\dagger}
\newcommand{\Aonetwohat}{\hat{A}_{1,2}}
\newcommand{\Aonetwohatd}{\hat{A}_{1,2}^\dagger}

\newcommand{\rhoideal}{\rho_{\mathrm{ideal}}}
\newcommand{\rhoactual}{\rho_{\mathrm{actual}}}
\newcommand{\rhoapprox}{\rho_{\mathrm{approx}}}
\newcommand{\rhoe}{\rho_{\mathcal{E}}}
\newcommand{\rhof}{\rho_{\mathcal{F}}}
\newcommand{\rhoef}{\rho_{\mathcal{E},\mathcal{F}}}
\newcommand{\gopt}{g_{\mathrm{opt}}}
\newcommand{\tgate}{t_{\mathrm{pulse}}}

\newcommand{\dchi}{d_{\chi}}
\newcommand{\drho}{d_{\rho}}
\newcommand{\Dactapprox}{\mathcal{D}^{\mathrm{inf}}_{|\Psiin\rangle}(\rhoactual,\rhoapprox)}

\SetKwInOut{Input}{Input}
\SetKwInOut{Output}{Output}
\SetKwProg{Fn}{Function}{:}{}
\SetKwFunction{FMainone}{ClusterExpansion}
\SetKwFunction{FMaintwo}{CodeSimulation}
\SetKwFunction{Fstringjoin}{string\_join}
\SetKwFunction{Fgetlastchar}{get\_last\_char}
\newtheorem{theorem}{Theorem}
\newtheorem{lemma}{Lemma}

\begin{document}

\title{Accurate and Honest Approximation of Correlated Qubit Noise}

\author{F. Setiawan}\email{setiawan.wenming@riverlane.com}
\affiliation{Riverlane Research Inc., Cambridge, Massachusetts 02142, USA}
\author{Alexander V. Gramolin}
\affiliation{Riverlane Research Inc., Cambridge, Massachusetts 02142, USA}
\author{Elisha S. Matekole}
\affiliation{Riverlane Research Inc., Cambridge, Massachusetts 02142, USA}
\author{Hari Krovi}
\affiliation{Riverlane Research Inc., Cambridge, Massachusetts 02142, USA}
\author{Jacob M. Taylor}\thanks{Present Address: Joint Center for Quantum Information and Computer Science,
University of Maryland-NIST, College Park, Maryland 20742, USA}
\affiliation{Riverlane Research Inc., Cambridge, Massachusetts 02142, USA}

\date{April 1, 2025}
\begin{abstract}
Accurate modeling of noise in realistic quantum processors is critical for constructing fault-tolerant quantum computers. While a full simulation of actual noisy quantum circuits provides information about correlated noise among all qubits and is therefore accurate, it is, however,  computationally expensive as it requires resources that grow exponentially with the number of qubits. In this paper, we propose an efficient systematic construction of approximate noise channels, where their accuracy can be enhanced by incorporating noise components with higher qubit-qubit correlation degree. To formulate such approximate channels, we first present a method, dubbed the cluster expansion approach, to decompose the Lindbladian generator of an actual noise channel into components based on interqubit correlation degree. 
We then generate a $k$-th order approximate noise channel by truncating the cluster expansion and incorporating noise components with correlations up to the $k$-th degree. We require that the approximate noise channels must be accurate and also ``honest", i.e., the actual errors are not underestimated in our physical models. As an example application,  we apply our method to model noise in a three-qubit quantum processor that stabilizes a $\dsl2,0,2\dsr$ codeword, which is one of the four Bell states. We find that, for realistic noise strength typical for fixed-frequency superconducting qubits coupled via always-on static interactions, correlated noise beyond two-qubit correlation can significantly affect the code simulation accuracy. Since our approach provides a systematic characterization of multi-qubit noise correlations, it enables the potential for accurate, honest and scalable approximations to simulate large numbers of qubits from full modeling or experimental characterizations of small enough quantum subsystems, which are efficient yet still retain essential noise features of the entire device.

\end{abstract}
\maketitle 

\indent

\section{Introduction}
Noise presents the major stumbling block in constructing a scalable quantum information processing device. To protect quantum information against noise such that error rates are small enough for practical quantum advantage in large-scale quantum processors, one needs to implement quantum error correction (QEC)~\cite{Shor1995Scheme,shor1996fault,Calderbank1996Good,devitt2013quantum}, where a logical qubit is encoded using many physical qubits. The basic idea behind QEC is that the logical error rate can be made arbitrarily small by increasing the number of physical qubits provided that the physical-qubit error rate is below a certain  threshold~\cite{aharonov1997fault,Gottesman1998Theory,preskill1998reliable,knill1998resilient}. Since the assessment of QEC performance is based on the assumption of the underlying noise model,  
accurate noise modelings and characterizations of quantum processors are absolutely essential for designing reliable and practical QEC. Moreover, faithful information on noise in realistic quantum devices is also crucial for devising more noise-resilient qubit architectures and hardware-efficient QEC codes with better threshold and smaller overhead. For example, for devices with biased Pauli noise, cat qubits~\cite{Guillard2019Repetition,Darmawan2021Practical,Chamberland2022Building} and the XZZX surface code~\cite{bonilla2021xzzx} are favorable as they possess outstanding resilience against such noise.

Knowledge about noise in quantum devices can be obtained by doing full simulations or experimental characterizations, e.g., gate set~\cite{blume2017demonstration,nielsen2021gate,proctor2020detecting}, quantum state and process tomography~\cite{paris2004quantum,howard2006quantum}, Lindblad tomography~\cite{Samach2022Lindblad} or randomized benchmarking~\cite{emerson2005scalable,Magesan2011Scalable}, of the actual physical device, which supply details on the qubit state at a given moment or quantum process/channel due to a gate application of a certain duration. Such theoretical simulations or experimental characterizations,  however, require computational or experimental resources (such as number of measurements) that scale exponentially with the number of qubits, which makes them impractical for large quantum circuits. 
The question that naturally arises is how one can use the full characterization of noise in small quantum processors to accurately assess the performance of large quantum information processors?

\begin{figure}[t!]
\includegraphics[width=\linewidth]{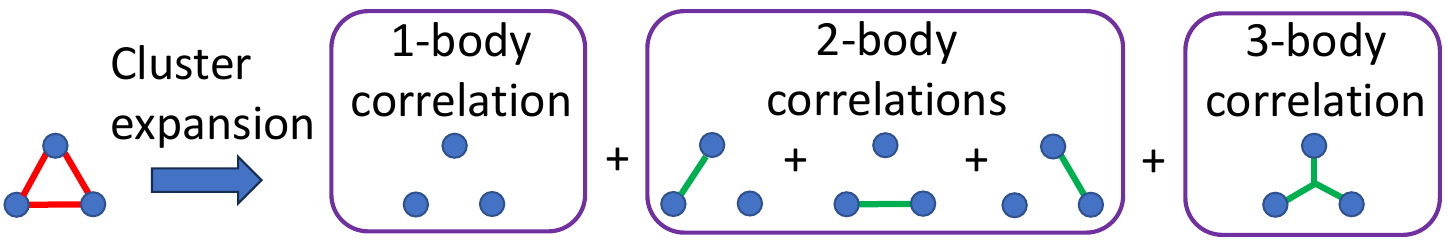}
	\caption{Schematics of the cluster expansion approach, a perturbative series method that decomposes the Lindbladian (noise generator) into terms based on the interqubit correlation degrees. Shown is an example of the cluster expansion approach applied to noise in a three-qubit quantum processor with triangular connectivity.}\label{fig:clusterexpansion}
 \end{figure}

One strategy used to simulate noise in large quantum systems is to assume a Pauli noise model~\cite{harper2020efficient,van2023probabilistic,Harper2023Learning} and construct approximate stochastic noise channels based on the Pauli twirling approximation~\cite{emerson2007symmetrized,Silva2008Scalable} which averages the noise channels with respect to a set of unitaries. This oversimplified model, however, often does not provide an accurate and honest description of noise in real quantum systems as it does not capture coherent noise, such as crosstalk, pulse miscalibration, etc., and non-Pauli errors, such as leakage, etc., which are the dominant noise in many realistic quantum devices. In this paper, we propose a general systematic approach to generate approximate noise channels beyond the Pauli noise model, which is based on actual noise data obtained from either realistic simulations or experimental characterizations of subsystems of the quantum device. In particular, we can obtain the actual noise data via either Lindbladian dynamics simulation or experimental characterizations such as Lindblad tomography~\cite{Samach2022Lindblad}, and process tomography~\cite{howard2006quantum} which supplies data that can be fitted to get the best-fit Lindbladian for the quantum channel~\cite{onorati2021fitting,onorati2023fitting}. In this paper, we focus on the former, i.e., using inputs from theoretical simulations of Lindbladian dynamics. 

We present an approach, dubbed the \textit{cluster expansion} method, to 
 decompose the actual noise into components according to their degree of qubit-qubit correlations [see Fig.~\ref{fig:clusterexpansion}]. A $k$-th order approximate noise channel is then generated by truncating the expansion and including noise components with qubit-qubit correlations up to the $k$-th degree. The accuracy of the approximate noise model can be improved (at the expense of reducing the efficiency) by including noise components with increasingly higher interqubit correlation degree. Further, we impose that our approximate noise channel must be constructed in an ``honest" manner, i.e., the approximate error does not underestimate the actual error.
Therefore, the aim of this paper is to provide a systematic method to construct  approximate noise models, $\Napprox$, such that
\begin{enumerate}
\item the action of $\Napprox$ can be made as “accurate” (faithful) as possible in representing the action of the actual channel $\Nactual$,
\item $\Napprox$ gives an “honest” representation of the noise in the actual system, and
\item the construction of $\Napprox$ for the full system can be done by combining together the input noise data from its smaller subsystems in a reliable way where there is no overcounting of any of the noise components.
\end{enumerate}

This paper is organized as follows. In Sec.~\ref{sec:Lindblad}, we begin by giving a short review of the Lindblad master equation and also define the normal-form of a noise channel. Subsequently in Sec.~\ref{sec:cluster}, we use the cluster expansion approach to decompose the Lindbladian into terms based on their qubit-qubit correlation degrees. In Sec.~\ref{sec:approximate}, we present ways to construct approximate noise channels using the decomposed Lindbladians and define the honesty and accuracy criteria of the approximate noise models. We then give a detailed application of our noise model for QEC in Sec.~\ref{sec:200code},  focusing on a specific example, namely the $\dsl 2,0,2\dsr$ code.  We give the conclusion and provide the outlook on future research directions in Sec.~\ref{sec:conclusions}. The appendices contain details on the superoperator representation, proof of the cluster expansion, example applications of the cluster expansion method, algorithms and details of pulse sequences for the $\dsl 2,0,2\dsr$ code, and results for simulating the $\dsl 2,0,2\dsr$ code in a three-qubit quantum processor with triangular connectivity.

\section{Lindblad noise model}\label{sec:Lindblad}
We assume that the evolution of a quantum processor interacting with its environment is well approximated by the Lindblad master equation~\cite{lindblad1976generators,gorini1976completely}: 
\begin{equation}\label{eq:Lindblad}
\frac{d\rho(t)}{dt} = \Lcal\rho(t),
\end{equation} 
where $\rho$ is the system density matrix and the Lindbladian $\Lcal$ is given by~\cite{lindblad1976generators,gorini1976completely}
\begin{equation}\label{eq:Lcal}
\Lcal(\rho) = -i[\Hhat,\rho] +\sum_{j}\gamma_j\left(\Ajhat \rho \Ajhatd -\frac{1}{2}\Ajhatd\Ajhat\rho - \frac{1}{2}\rho\Ajhatd\Ajhat \right),
\end{equation} 
with $[\cdot,\cdot]$ denoting the commutator. Here the Hamiltonian $\Hhat$ describes the unitary dynamics of the system, which can include coherent errors resulting from pulse miscalibration, crosstalk due to always-on qubit-qubit interactions, etc. $\Ajhat$ is the $j$th Lindblad (jump) operator with a jump/dissipation rate $\gamma_j$ representing the incoherent noise (e.g., dephasing and relaxation) due to interactions between the system and environment. The Lindblad master equation can describe either Markovian or non-Markovian noise~\cite{groszkowski2023simple,fleming2012non,Gulacsi2023signatures}. The resulting channel (propagator) that evolves the system for a duration $t$ can be written as
\begin{equation}\label{eq:Vt}
\Vt = \hat{T}\left[\exp \left(\int_{0}^{t}\mathcal{\Lcal}(t')dt' \right)\right],
\end{equation}
where $\hat{T}$ denotes the time-ordering operator.

In order to describe dynamics due to only the unwanted noise, we write the channel (propagator) in the \textit{normal form}. This normal form is defined as the noise channel with the unitary dynamics contribution $\Ucald$ of the ideal Hamiltonian being factored out, i.e.,
\begin{equation}\label{eq:Ttt}
\Nt =  \Vt \circ \Ucald,
\end{equation}
where $\Ucald(\rho) = U^\dagger\rho U$ is the ideal (noiseless) channel with $U$ being the ideal target unitary operator~\footnote{We remark that, given the target $U$, there exists a family of such forms
\begin{equation}
\mathcal{N}_x = \mathcal{U}^{-x} \circ\mathcal{V} \circ \mathcal{U}^{-1+x}
\end{equation}
for $x \in [0,1]$ with $\mathcal{U}^x(\rho) = U^x \rho U^{-x}$. In addition, summation over a function $f(x) dx$ also produces appropriate forms, e.g., $\mathcal{N}_f  = \int \mathcal{N}_x f(x) dx$ with $\int f(x) dx = 1$ being the only requirement to ensure that $\mathcal{N}_f$ is the identity channel when there is no noise. There can be circumstances where the use of such $\mathcal{N}_f$ forms may be desirable, due to averaging or simplification for the appropriate setting, but here we fix $\mathcal{N}_0$ as the \emph{normal} form.}.

If the Hamiltonian or Lindblad operators are time-dependent then the Lindbladian $\Lcal$ and hence, the noise channel are also time-dependent. We capture the effect of this time-dependent dynamics by using $\Lcalt$, which is an effective dimensionless time-independent Lindbladian that serves as a generator for the system evolution from the beginning to the end of the application of the noise channel. In the superoperator representation~\cite{navarrete2015open} (see Appendix~\ref{sec:superop}), it is given by  
\begin{equation}\label{eq:Lcaleff}
\Lcalthat = \mathrm{ln}(\Nthat),
\end{equation}
where $\Lcalthat$ is the superoperator representation of the Lindbladian $\Lcalt$ and $\Nthat$ is the superoperator representation of the normal-form noise channel $\mathcal{N}$ as defined in Eq.~\eqref{eq:Ttt}. In general, the matrix logarithm in Eq.~\eqref{eq:Lcaleff} yields an infinite number of branch cuts. However, for small enough noise strength, the normal-form noise channel is close to the identity channel and hence the zeroth branch cut of the matrix logarithm taken in Eq.~\eqref{eq:Lcaleff} gives a physical Lindbladian, i.e., the Lindbladian where all its eigenvalues have non-positive real parts.

\section{Lindbladian decomposition using cluster expansion}\label{sec:cluster}
Quantum circuits are usually transpiled into one- and two-qubit gates. Noise associated with the gate operations are then expected to have significant effects only in the neighborhood of the qubits where the gates are applied. As a result, we can simplify the noise model by first decomposing the noise into locally correlated noise and incorporating only the subset of noise that substantially affects the gates performance.

To separate and quantify the contributions due to noise with different degree of correlations, we present a 
 method akin to the cluster expansion approach~\cite{mayer1941molecular,kira2011semiconductor} used for  interacting closed many-body systems. The cluster expansion approach is a perturbative series expansion which sums over the contributions from clusters with different size; see Fig.~\ref{fig:clusterexpansion}. Instead of doing a power-series expansion over the partition function as in the standard cluster expansion approach used to treat interactions in closed many-body systems~\cite{mayer1941molecular,kira2011semiconductor}, our cluster expansion approach decomposes the Lindblad generator $\Lcalonent$ of an actual noise channel into sum of   components with different degree of qubit-qubit correlations, i.e.,
\begin{equation}\label{eq:clusterexpansion}
\Lcalonent = \Lcalonenonet + \Lcalonentwot + \cdots +\Lcalonennt,
\end{equation}
where $\Lcalonent$ denotes the exact Lindbladian of the whole system which comprises $n$ number of qubits. Here, 
\begin{align}\label{eq:decomposition}
\Lcalonenmt &= \sum_{S_m } \Lcalonent^{S_m}
\end{align}
is the component of the Lindbladian that describes noise with $m$-th correlation degree. The sum used to evaluate the $m$-th order correlated Lindbladian $\Lcalonenmt$ in Eq.~\eqref{eq:decomposition} is performed over all subsets of size $m$ in the set $S= \{1,2,\cdots, n\}$, i.e., $ S_m  \subseteq S$ such that $|S_m| = 
m$. Each of the $m$-th order Lindblad terms is given by the following recurrence relation:
\begin{equation}\label{eq:Lsmdef}
 \Lcalonent^{S_m} = \left( \frac{1}{d^{|\bar{S}_m|}}\mathrm{Tr}_{\bar{S}_m} \left[\Lcalonent \right]\otimes \mathds{1}_{\bar{S}_m} \right) - \sum_{R\subsetneq S_m} \Lcalonent^{R}. 
\end{equation}
  Here $\bar{S}_m = \{x \in S: x  \notin S_m\}$ denotes the complement subset of $S_m$, where $S_m$ is the subsystem containing $m$ qubits on which the noise correlation is calculated. As proven in Theorem~\ref{theorem_one} of Appendix~\ref{sec:proofrecur}, the above definition of $\Lcalonent^{S_m}$ [Eq.~\eqref{eq:Lsmdef}] guarantees that any Lindbladian can always be decomposed in terms of $\Lcalonent^{S_m}$ as in Eqs.~\eqref{eq:clusterexpansion} and~\eqref{eq:decomposition}. Moreover, Eq.~\eqref{eq:Lsmdef} also ensures that the term $\Lcalonent^{S_m}$ indeed describes a pure $m$-th order noise correlation as it can be decomposed using identities on all qubits and generators that have non-trivial basis only on all of the qubits $j \in S_m$ (see Theorem~\ref{theorem_two} of Appendix~\ref{sec:proofrecur} for a proof). We note that our cluster expansion for decomposing the Lindbladian is general and can be applied to any quantum processes, even though in this paper we apply our cluster expansion only to a specific quantum process represented by a QEC circuit.

As shown in Eq.~\eqref{eq:recur}, $\Lcalonent^{S_m}$ is calculated by first doing a partial trace of the Lindbladian of the whole $n$-qubits system over the $(n-m)$ qubits in the complement subset $\bar{S}_m$ which live in the Hilbert space with dimension $d^{(n-m)}$.  Here, $d$ is the number of energy levels in the qubit, where we use $d = 2$ for the simulation in this paper.  The tensor product with the identity $\mathds{1}_{\bar{S}_m}$ in the $(n-m)$-qubit subspace after the trace operation is to ensure that the dimension of the Lindblad component is the same as its initial dimension. The result after the partial trace gives an effective Lindbladian that describes noise in an $m$-qubits subsystem where it includes all noise components with correlation degrees up to the $m$-th order. Therefore, to get the $m$-th order correlated noise term $\Lcalonent^{S_m}$, we need to subtract all the lower ($m'<m$)-th order correlated noise terms from the partial-trace result, where the subtraction is over all proper subsets of $S_m$ ($R\subsetneq S_m$), i.e., all  subsets of $S_m$ that are not equal to $S_m$. 

In order to connect to experiments where noise characterization is done by evolving the system through a quantum channel $e^{\Lcalonent}$, we need to express $\Lcalonent^{S_m}$ in terms of $e^{\Lcalonent}$. In the small-noise regime where $||\Lcalonent|| \leq \varepsilon < 1/2$, we can approximate $e^{\mathcal{\Lcalonent}} \approx \mathds{1} + \Lcalonent$ and $\mathrm{ln} (\mathds{1}+ \Lcalonent) \approx \Lcalonent$. We note that the error in making these approximations scales only with $\varepsilon^2$, which is of higher order than the magnitude of the noise $||\Lcalonent||$ that we would like to approximate. Therefore, we can rewrite Eq.~\eqref{eq:Lsmdef} as (see Appendix~\ref{sec:proofrecur} for a rigorous proof):
\begin{align}\label{eq:recur}
\Lcalonent^{S_m}& \approx \mathrm{ln}\left( \mathrm{Tr}_{\bar{S}_m} \left[e^{\Lcalonent }\left(\mathds{1}_{S_m}\otimes\rho_{\bar{S}_m} \right)\right]\otimes \mathds{1}_{\bar{S}_m} \right) - \sum_{R\subsetneq S_m} \Lcalonent^{R} \nonumber\\
&\approx \mathrm{ln}\left( \mathrm{Tr}_{\bar{S}_m} \left[e^{\Lcalonent }\left( \mathds{1}_{S_m}\otimes\frac{\mathds{1}_{\bar{S}_m}}{{d}^{(n-m)}} \right)\right]\otimes \mathds{1}_{\bar{S}_m} \right)   \nonumber\\
&\hspace{0.5cm}- \sum_{R\subsetneq S_m} \Lcalonent^{R}.
\end{align}
Here, the density matrix $\rho_{\bar{S}_m}$ of the ($n-m$) qubits in the subset $\bar{S}_m$, the complement subset to the qubit subset of interest, is taken to be a fully mixed state; this is a valid assumption for QEC circuits, the focus of our paper, as noise in data qubits makes the syndrome qubits undergo ergodic evolutions into all possible states.  To further add to the understanding and for the ease of implementation, we provide an algorithm to compute the recurrence relation of the cluster expansion in Algorithm~\ref{algo_recur}.
\begin{algorithm}\small
\caption{$m$-th-order Lindbladian $\Lcalonent^{S_m}$ extraction from a quantum channel $\Nthat$}\label{algo_recur}
\Input{\begin{enumerate}
\item A quantum channel  $\Nthat$
\item an integer $n$ = total number of qubits in the system
\item an integer $m$ = number of qubits in the subsystem $S_m$ of interest
\item $m$ integers: $S_m = \{q_1, \cdots, q_m\}$ denoting qubit indices in the subset $S_m$ of interest
\item an integer $d$ = number of levels in the qubit
\end{enumerate}}
\Output{$\Lcalonent^{S_m}$: the Lindbladian term describing a pure $m$-th body correlation in the $S_m$ subsystem consisting of $m$ ($m<n$) qubits.}
\Fn{\FMainone{$\Nthat$, $n$, $m$, $S_m$, $d$}}{
\For{$i \gets 1$ \KwTo  $m$}{
Form all $R_i \subseteq S_m$ where $|R_i| = i$\\
\tcp{Initialize $\Lcalthat^{R_i}$ by partial-tracing $\Nthat$ over $\bar{R}_i$}
$\Lcalthat^{R_i} \gets \mathrm{ln}\left(\frac{1}{d^{n-i}}\mathrm{Tr}_{\bar{R}_i} \left[\Nthat\right]\otimes \mathds{1}_{\bar{R}_i} \right)$
\\
\For{$Q\subsetneq R_i$}{
\tcp{Subtract all the lower order Lindblad terms}
$ \Lcalthat^{R_i} \gets \Lcalthat^{R_i}- \Lcalthat^Q$}
Store $\Lcalthat^{R_i}$
}
\Return $\Lcalthat^{S_m}$}
\end{algorithm}

\section{Approximate noise model}\label{sec:approximate}
\subsection{Construction}\label{sec:construction}
Using the above decomposed Lindbladians, we can now ask the question: can we truncate the total Lindblad term $\Lcalonent$ by considering only terms up to the $k$-th order, where $k \leq n$? As we show in the rest of this work, we answer this in the affirmative. Specifically, we consider the $k$-body approximation of the Lindbladian $\Lcalonekt$ by ignoring all higher-order terms ($m>k$), defining:
\begin{equation}\label{eq:truncate}
\Lcalonekt = \Lcalonenonet + \Lcalonentwot + \cdots +\Lcalonenkt.
\end{equation}
We remark that, in general, each term in $\Lcalonekt$ is not guaranteed to be a physical Lindbladian. This means that $\Lcalonekt$ may also not be a physical Lindbladian since it can comprise a sum of both positive and negative Lindblad terms [see Eqs.~\eqref{eq:decomposition} and~\eqref{eq:recur}]~\footnote{While the term $\Lcalonent^{S_m}$ in Eq.~\eqref{eq:recur} is not guaranteed to be a physical Lindbladian, the first term of the right-hand side of Eq.~\eqref{eq:recur}, i.e.,  $\mathrm{ln}\left( \mathrm{Tr}_{\bar{S}_m} \left[e^{\Lcalonent }\left(\mathds{1}_{S_m}\otimes\rho_{\bar{S}_m} \right)\right]\otimes \mathds{1}_{\bar{S}_m} \right)$, is always a physical Lindbladian as it describes noise in an $m$-qubits subsystem.}. However, all the 1-body terms $\Lcalonent^{S
_1}$ and hence $\Lcalonenonet$ are guaranteed to be physical Lindbladians. Note that the way we calculate each term in $\Lcalonekt$, as defined in Eqs.~\eqref{eq:decomposition} and~\eqref{eq:recur}, guarantees that the summation in Eq.~\eqref{eq:truncate} does not over- or under-count the expected error, as we showcase with several examples in Appendix $\ref{app:examples}$. 

To generate the approximate noise channel, we can in principle exponentiate the $k$-body approximate Lindbladian $\Lcalonekt$. However, for sufficiently small noise strength, we can further simplify the computation for the approximate noise channel by making a Trotter approximation. Combining the  truncation of the Lindbladian cluster-expansion and the Trotter approximation, we write the $k$-th order approximate noise channel as
\begin{align}\label{eq:Napproxm}
&\Napproxhatk(g)\equiv \prod_{S_1} \exp(g\Lcalonent^{S
_1}) \cdots \prod_{S_k}\exp(g\Lcalonent^{S_k}),
\end{align}
where $g$ is the gain factor~\footnote{In general, one may have to use different gain factors for different decomposed components of the Lindbladian to get a completely-positive and trace-preserving (CPTP) approximate noise channel that is honest and accurate. However, for our case of [[2,0,2]] code in 3-qubit processors, we find that we can use the same gain factor for all the Lindbladian components to obtain honest and accurate CPTP approximate noise channels.}, a free parameter to be optimized in order to make the approximation ``honest" and as accurate as possible (see Sec.~\ref{sec:honesty} for details on the honesty and accuracy criteria). We use this Trotterized form in what follows as it allows us to reason about noise in the context of more complex quantum circuits, where each noise component can be replaced by a quantum operation occurring over a subset $S_m$ with $m \leq k$. 

We note that our approximate noise model can be scaled to larger systems even when there is crosstalk between all connected pairs of qubits as long as the effect due to the crosstalk between a pair of qubit does not extend across the whole system and is significant only within a region near the pair of qubit where the crosstalk acts. This is because the cluster expansion can be truncated at a certain depth where the small amplitude of higher-order terms, i.e.,  terms with  correlation degree greater than $k$, allows us to entirely neglect these higher-order residual errors. 

\subsection{Honesty and accuracy}\label{sec:honesty}
To quantify the performance of the approximate noise model, we calculate the distances between ideal (noiseless), actual, and approximate noise channels, which quantify the distinguishability between the channels. Based on these distances, we can characterize two properties of the approximate noise channels: honesty and accuracy~\cite{Magesan2013Modeling,Puzzuoli2014Tractable,Mauricio2015Comparison}. 

An approximate noise channel is honest if it does not underestimate the deleterious effect of the actual noise. On the other hand, the accuracy of an approximate channel refers to how closely it can emulate the impact of the actual noise. More explicitly, if we consider the ideal (noiseless) $\Videal$, actual $\Vactual$, and approximate  $\Vapprox$ target channels map an initial pure state $|\Psiin\rangle$ into the $\rhoideal$, $\rhoactual$, and $\rhoapprox$ density matrices, respectively, we can then define the honesty and accuracy criteria of the approximate noise channel $\Vapprox$ in terms of the distance metrics $\Dcal$ between these density matrices as  
\begin{align}
 \mathrm{honest\,if\hspace{0.2cm}} &\frac{\Dcal(\rhoideal,\rhoapprox)}{\Dcal(\rhoideal,\rhoactual)} \geq 1, \label{eq:honesty}\\
\mathrm{accurate\,if\hspace{0.2cm}\,}&\frac{\Dcal(\rhoideal,\rhoactual)}{\Dcal(\rhoactual,\rhoapprox)} \geq 1,  \label{eq:accuracy}
\end{align}
for every input pure state $|\Psiin\rangle$. 

For ease of understanding, we show in Fig.~\ref{fig:distance} the schematics of the honesty and accuracy criteria of approximate noise models. For the approximate noise model to be honest, its density matrix must lie outside the ball whose radius is defined by the distance between $\rhoideal$ and $\rhoactual$ (the shaded ball shown in Fig.~\ref{fig:distance}). This means that an honest approximate noise model should not underestimate the actual noise.  On the other hand, for the approximate noise model to be accurate, it must be as close as possible to the actual model; in other words, the distance between them (the green dashed line in Fig.~\ref{fig:distance}) must be as small as possible. This implies that for an approximate noise model to be accurate, the error of its approximation should not be larger than the noise that it approximates. Consequently, as shown in Fig.~\ref{fig:distance}, accurate and honest approximate noise models must lie on the green arc outside the shaded ball but close to the actual model.

\begin{figure}[t!]
\centering
\includegraphics[width=0.4\linewidth]{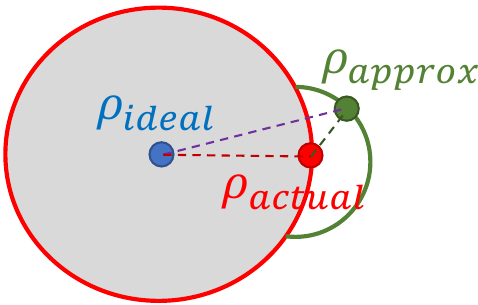}
	\caption{A schematic of distances between density matrices obtained from the ideal ($\rhoideal$), actual ($\rhoactual$) and approximate ($\rhoapprox$) models. Shaded area represents parameter regimes where  approximate noise models are not honest. Honest approximate noise models must lie outside the ball  defined by the distance between the ideal and actual models (shown by shaded ball). For the approximate noise models to be accurate, they must be close to the actual model. }\label{fig:distance}
\end{figure} 

To obtain the best approximate noise model, we find the optimal gain factor ($g = \gopt$) in $\Napproxhat$ [Eq.~\eqref{eq:Napproxm}] that yields the highest state-averaged accuracy ratio, i.e.,
\begin{equation*}
 \max_g\left\langle \frac{D(\rhoideal,\rhoactual)}{D(\rhoactual,\rhoapprox(g))}\right\rangle_{|\Psiin\rangle},
\end{equation*}
subject to the honesty constraint [Eq.~\eqref{eq:honesty}]. Here, $\langle \cdots \rangle_{|\Psiin\rangle}$ denotes the averaging over all the pure input states (for the case of QEC, which is the focus of our paper, these input states are all the eigenstates of the stabilizers).

While in this paper we choose to calculate the honesty and accuracy metrics by averaging over all pure input states which are the eigenstates of the stabilizers, for computational efficiency, one can instead choose a more scalable definition of honesty and accuracy metrics. For instance, we can use metrics which are defined using only the codewords of the QEC code (normally defined as the +1 eigenstates of the stabilizers). As an example, for the $\dsl 2,0,2 \dsr$ code which is the topic of the following sections, the codeword would just be the Bell state $|\Phi^+\rangle = (|00\rangle + |11\rangle)/\sqrt{2}$. This choice of honesty and accuracy metrics will give an exponential reduction in the amount of computation needed since the number of input states that needs to be taken into account is reduced exponentially. Furthermore, we can make the computation more efficient by using the honesty and accuracy ratios defined in terms of distances between syndrome measurement probability distributions. The reason we choose to use the honesty and accuracy metrics defined in terms of density matrices is because density matrices provide more information and also to conform with the way these two metrics are usually calculated in the literature, e.g., Refs.~\cite{Puzzuoli2014Tractable,Mauricio2015Comparison}. We emphasize that the applicability of our cluster expansion technique is not restricted to a particular definition of honesty and accuracy such as the one defined here. For computational efficiency, one can use a more scalable definition of honesty and accuracy metrics.

\begin{figure}
\includegraphics[width=\linewidth]{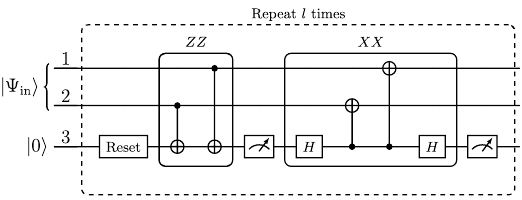}
  \caption{Quantum circuit for implementing the $\dsl2,0,2\dsr$ code. The code consists of two data qubits (qubits 1 and 2) and one syndrome qubit (qubit 3). The input state of the data qubits $|\Psiin\rangle$ is one of the Bell states while the syndrome qubit is initialized in state $|0\rangle$. The syndrome qubit is reset to $|0\rangle$ at the beginning of each round of the implementation of the $\dsl2,0,2\dsr$ code. The code has 2 stabilizers: $ZZ$ and $XX$. The $ZZ$ and $XX$ syndrome measurement values of each input Bell state $|\Psiin\rangle$ for each round of the implementation of an ideal (noiseless) $\dsl2,0,2\dsr$-code circuit are shown in Table~\ref{tab:insynd}.}
  \label{fig:200code}
\end{figure}

\begin{table}[t]\small
\centering
\renewcommand{\arraystretch}{2.2}
\begin{tabular}{|p{3.5cm} | p{3cm} |} 
 \hline\hline
 Input state $|\Psi_{\mathrm{in}}\rangle$ & ($ZZ$, $XX$) syndrome\\
 \hline 
$|\Phi^{+}\rangle = \dfrac{|00\rangle+|11\rangle}{\sqrt{2}}$ & (0,0)  \\ 
$|\Phi^{-}\rangle = \dfrac{|00\rangle-|11\rangle}{\sqrt{2}}$ & (0,1)  \\
$|\Psi^{-}\rangle = \dfrac{|01\rangle-|10\rangle}{\sqrt{2}}$ & (1,0)  \\
$|\Psi^{+}\rangle = \dfrac{|01\rangle+|10\rangle}{\sqrt{2}}$ & (1,1)  \\ 
\hline\hline
\end{tabular}
\caption{Data-qubits input states together with the syndrome measurement results of the ancilla qubit at each round of the implementation of an ideal (noiseless) $\dsl2,0,2\dsr$-code circuit.}\label{tab:insynd}
\end{table}

\section{Application to QEC: $\dsl2,0,2\dsr$ code}\label{sec:200code}
\subsection{Basics of $\dsl2,0,2\dsr$ code}
As an example application of our noise model, we consider applying it to the simplest QEC code, the $\dsl2,0,2\dsr$ code~\cite{corcoles2015demonstration}. As its name suggests, this code has 2 physical data qubits, zero logical qubits and a distance of 2. It has zero logical qubits since the number of stabilizers ($ZZ$ and $XX$) is the same as the number of physical data qubits. The code has only a single codeword which is any one of the four Bell states (see Table~\ref{tab:insynd}), which resides in a particular codespace with a fixed value (usually taken to be +1 eigenvalue) of the $ZZ$ and $XX$ parity operators. An arbitrary single-qubit error simply changes the codeword into another Bell state which is the opposite eigenstate of the $ZZ$ and/or $XX$ parity operators and thus can be detected via the $ZZ$ and $XX$ parity check measurements. For example, a bit-(phase-)flip changes the data-qubits' state to the eigenstate of $ZZ$ ($XX$) parity operator with an opposite eigenvalue, and a simultaneous bit- and phase-flip (a $Y$ error) changes it to the eigenstate with opposite eigenvalues for both the $ZZ$ and $XX$ parity checks.

The circuit for implementing the $\dsl2,0,2\dsr$ code is shown in Fig.~\ref{fig:200code}. Our circuit is simpler than the one implemented in Ref.~\cite{corcoles2015demonstration} where we use only one instead of two ancilla/syndrome qubits; in our circuit, the same ancilla qubit is used sequentially for the $ZZ$ and $XX$ parity checks. The input state of the circuit is taken to be any one of the four Bell states, where each of them has a unique syndrome measurement at each round of the implementation of an ideal (noiseless) $\dsl2,0,2\dsr$-code circuit, as shown in Table~\ref{tab:insynd}.

\subsection{Infidelity distance measure between different quantum channels}
Since noise in different quantum channels change the same initial data-qubits states into different quantum states, we can characterize noise in different quantum channels by first defining the fidelity between the data-qubits states resulting from two different quantum channels $\mathcal{E}$ and $\mathcal{F}$ as~\cite{jozsa1994fidelity,nielsen2010quantum}
\begin{equation}\label{eq:fidelity quantum}
F_{\mathrm{q}}(\rhoe,\rhof) = \left[\mathrm{Tr}\left( \sqrt{\sqrt{\rhoe}\rhof\sqrt{\rhoe}}\right)\right]^2.
\end{equation}
As different data-qubits states give rise to different probabilities of the syndrome-qubit measurement results $x$, we can similarly define the fidelity of the syndrome qubit measurements of different quantum channels in terms of their syndrome measurement probabilities [$\pe(x)$ and $\pf(x)$] as~\cite{nielsen2010quantum}
\begin{equation}
F_{\mathrm{c}}(\pe,\pf) = \left[\sum_{x}\sqrt{\pe (x) \pf (x)}\right]^2.
\end{equation}
Combining the fidelity measures characterized in terms of the data-qubits states and syndrome qubit measurements, we can define the infidelity distance between any two quantum channels used to implement QEC codes as
\begin{widetext}
\begin{align}\label{eq:infidelity}
\Dinfinput(\rhoe,\rhof) = 1-\left(\sum_{x\in \{0,1\}^{2l}}\bigg[\sqrt{\pe(x|\Psiin) \pf(x|\Psiin)}\times\mathrm{Tr}\left( \sqrt{\sqrt{\rhoe(x|\Psiin)}\rhof(x|\Psiin)\sqrt{\rhoe(x|\Psiin)}}\right)\bigg]\right)^2,
\end{align}
\end{widetext}
where $x \in \{0,1\}^{2l}$ are the syndrome measurement bit strings after the $l$-th round implementation of the $\dsl2,0,2\dsr$-code circuit and $|\Psiin\rangle$ is the data-qubits input state, which is one of the four Bell states. Here, $\rhoef(x|\Psiin)$ are the data-qubits density matrices resulting from the application of the quantum channels $\mathcal{E}$ and $\mathcal{F}$, respectively,  given the syndrome measurement result $x$ and a particular input state $|\Psiin\rangle$. The density matrices $\rhoef(x|\Psiin)$ are therefore associated with the corresponding probability distributions $\pef(x|\Psiin)$ of a syndrome measurement result $x$ given a particular input state $|\Psiin\rangle$. The infidelity values are given by $0 \leq \Dinfinput(\rhoe,\rhof) \leq 1$, where $\Dinfinput(\rhoe,\rhof) = 0$ means that the two density matrices are identical and  $\Dinfinput(\rhoe,\rhof) = 1$ indicates that the two density matrices are furthest apart from each other. 

In the context of QEC, the honesty and accuracy criteria [Eqs.~\eqref{eq:honesty} and~\eqref{eq:accuracy}] of our approximate noise models can then be defined in terms of the infidelity metrics given in Eq.~\eqref{eq:infidelity}. Specifically, we require these criteria to be satisfied for every input pure state $|\Psiin\rangle$ belonging to the eigenstates of the stabilizers (which are the four Bell states for our example of the $\dsl 2,0,2 \dsr$ code).

\subsection{Simulation}\label{sec:simulation}
We now analyze the implementation of our noise model for the $\dsl2,0,2\dsr$ code in a concrete physical 
 setup: a quantum processor consisting of fixed-frequency qubits coupled to each other via static transverse couplings. The system Hamiltonian is given by
\begin{equation}\label{eq:Htot}
\Hhat(t) = \Hq+ \Hint +\Hdr(t),
\end{equation}
where $\Hq$ is the bare qubit Hamiltonian, $\Hint$ is the qubit-qubit interaction term and $\Hdr(t)$ is the qubit drive term.

We consider the qubits to be a collection of two-level systems with the bare Hamiltonian (in units $\hbar = 1$)
\begin{equation}\label{eq:Hq}
\Hq = \sum_{j} \frac{\omegaqj}{2}\sigmazj,
\end{equation}
where $\sigmazj = |0\rangle_j {}_j\langle 0| - |1\rangle_{j} {}_j\langle 1|$ is the $j$-th qubit Pauli-$z$ operator and $\omegaqj$ is the $j$-th qubit frequency. The qubits are coupled to each other via always-on static transverse couplings with the Hamiltonian 
\begin{equation}\label{eq:Hint}
\Hint = \sum_{\langle j,k \rangle} J_{jk} (\sigmapj\sigmamk +\sigmapk\sigmamj),
\end{equation}
where the sum is over all pairs of interacting qubits. Here, $J_{jk}$ is the strength of the coupling between qubits $j$ and $k$, $\sigmap = |1\rangle\langle 0|$ and $\sigmam = |0\rangle\langle 1|$ are the qubit raising and lowering operators, respectively.  We consider two different connectivities: linear and triangle geometries (see Fig.~\ref{fig:geometry}), where the linear geometry has one less qubit-qubit coupling compared to the triangle geometry. 

\begin{figure}[t]
\centering
\includegraphics[width=0.6\linewidth]{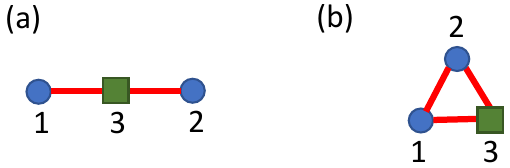}
	\caption{Connectivities of the 3-qubit processor used for implementing the $\dsl 2,0,2\dsr$ code: (a) Linear and (b) Triangle geometry. Circles represent data qubits and squares represent syndrome qubits. Numbers indicate the qubit labels as in Fig.~\ref{fig:200code}.}\label{fig:geometry}
\end{figure}

Gates are implemented by driving the qubits with control fields where the drive Hamiltonian is
\begin{align}\label{eq:Hdr}
\Hdr(t) &= \sum_{\substack{j\in \mathrm{driven}\\ \mathrm{qubits}}}\Omegadrj(t) \left(\sigmapj e^{-i(\omegadrj t+\phidrj)} + \mathrm{H.c.}\right).
\end{align}
Here $\Omegadrj(t)$, $\omegadrj$ and $\phidrj$ are respectively the envelope, frequency and phase of the control field that drives qubit $j$. The notation ``H.c.'' denotes the Hermitian conjugate of the preceding term. A single-qubit gate is implemented by driving the qubit with a control field at the qubit's characteristic frequency. Any single-qubit gate can be implemented in this way by choosing appropriate amplitude, duration and phase of the driving pulse. However, here we choose to realize single-qubit gates using sequences of $X$ and $Y$ rotations obtained using the $X$-$Y$ decomposition~\cite{nielsen2010quantum} of the gates. The two-qubit (CNOT) gate used in our simulation is a cross-resonance~\cite{Rigetti2010Fully,Chow2011Simple,Cruz2021Testing,cruz2023shallow} gate which is implemented by driving the control qubit at the frequency of the target qubit. To cancel the unwanted effect due to crosstalk in the cross-resonance gate, we apply an echoed version of the calibrated pulse sequence as described in Ref.~\cite{Sheldon2016Procedure}. 

\begin{figure*}
\includegraphics[width=\linewidth]{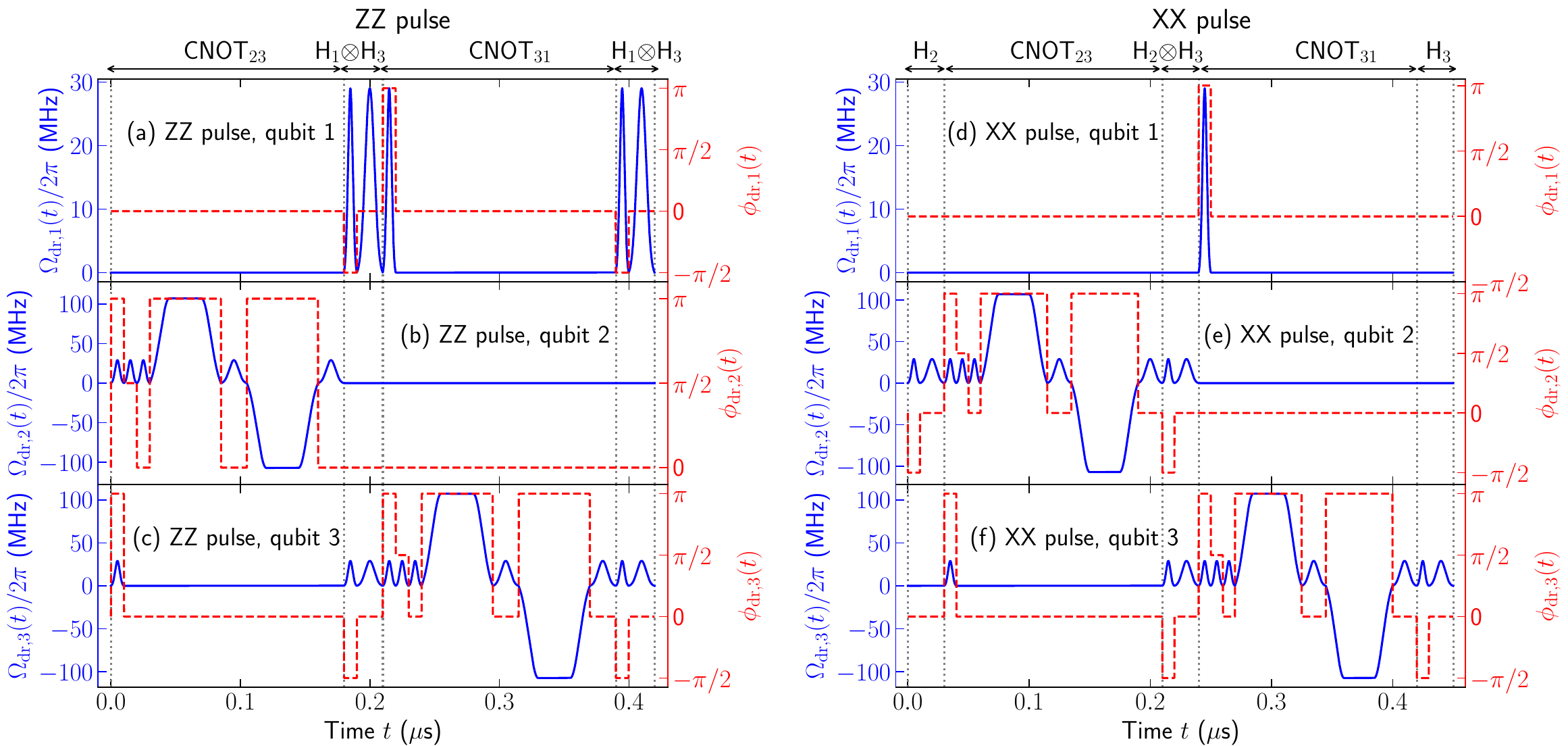}
	\caption{Amplitudes $\Omega_{\mathrm{dr},j}(t)$ and phases $\phi_{\mathrm{dr},j}(t)$ of the control pulses for the $ZZ$ [panels (a), (b), and (c)] and $XX$ [panels (d), (e), and (f)] stabilizers of the $\dsl2,0,2\dsr$ code implemented in a linear geometry. From top to bottom: pulse schedules for qubit $j=1,2,3$ (the qubit labels are the same as those in Fig.~\ref{fig:200code}). The pulse sequences consist of H$_a$ and CNOT$_{ab}$ gate pulses where H$_a$ denotes the Hadamard on qubit $a$ and CNOT$_{ab}$ denotes the CNOT gate with qubit $a$ being the control qubit and qubit $b$ being the target qubit. The Hadamard and CNOT gates in turn comprise two or more elementary pulses as discussed in Appendix~\ref{sec:pulse_sequence}.}\label{fig:pulse}
\end{figure*}

The pulse schedules of the $ZZ$ and $XX$ stabilizers, which consist of single- and two-qubit gate pulses, are shown in the left and right panels of Fig.~\ref{fig:pulse}, respectively; here we use Gaussian and Gaussian-square pulse envelopes for the single-qubit and two-qubit gates, respectively (see Appendix~\ref{sec:pulse_sequence} for details on the pulse sequences for the $ZZ$ and $XX$ parity checks). These pulse schedules are for a linearly-connected qubit geometry (the pulse schedules for the triangle geometry are not shown as they are similar to  those for the linear geometry). They are calibrated for the qubit parameters shown in Table~\ref{tab:params}; these parameters are typical for fixed-frequency superconducting qubits. For the linear geometry, our calibrated pulses yield average fidelities of $0.9896$ and $0.9888$ for the $ZZ$ and $XX$ stabilizers, respectively. Similar calibrated pulses for the triangle geometry yield average fidelities of $0.9870$ and $0.9832$ for the $ZZ$ and $XX$ stabilizers, respectively.

\renewcommand{\arraystretch}{0.9}
\begin{table}\small
\centering
\begin{tabular}{|p{3 cm} | p{4.5 cm}|} 
 \hline\hline
 Qubit frequencies &  Qubit-qubit couplings  \\  
 \hline
$\omega_{q_1}/2\pi$ = 4.8 GHz  & $J_{1,2}/2\pi$ = 4 MHz (Triangle)\\ 
& \hspace{1.05 cm} = 0 (Linear)\\
$\omega_{q_2}/2\pi$ = 5.2 GHz & $J_{2,3}/2\pi$ = 4 MHz\\ 
$\omega_{q_3}/2\pi$ = 5 GHz & $J_{1,3}/2\pi$ = 4 MHz\\ 
 \hline
Relaxation times & Dephasing times \\ 
\hline 
$T_{1,1}$ = 100 $\mu$s  & $T_{\varphi,1}$ = 200 $\mu$s\\ 
$T_{1,2}$ = 100 $\mu$s & $T_{\varphi,2}$ = 200 $\mu$s\\ 
$T_{1,3}$ = 100 $\mu$s & $T_{\varphi,3}$ = 200 $\mu$s\\ 
 \hline\hline
\end{tabular}
\caption{Qubit parameters: qubit frequencies $\omega_{q}$, qubit-qubit coupling strengths $J$, relaxation times $T_1$ and pure dephasing times $T_{\varphi}$.}\label{tab:params}
\end{table}

To simulate the gate dynamics which incorporates the effects of dissipation due to $T_{\varphi}$ dephasing  and $T_1$ relaxation, we use the Lindblad master equation:

\begin{align}\label{eq:Lindbladappendix}
\frac{d\rho(t)}{dt} &= -i[\Hhat(t),\rho(t)] + \sum_{j}  \frac{1}{\Tphij} \bigg(\sigmazj \rho(t) \sigmazj  - \rho(t)\bigg) \nonumber\\
&\hspace{0.2cm}+ \sum_j \frac{1}{\Tonej}\left(\sigmamj \rho(t) \sigmapj -\frac{1}{2}\left\{\sigmapj\sigmamj,\rho(t)\right\} \right) ,
\end{align}
where $\rho(t)$ is the density matrix of the whole system, $\Tphij$ and $\Tonej$ are the dephasing and relaxation times of qubit $j$, respectively. Here $\{\cdot,\cdot\}$ denotes the anticommutator. The Lindblad master equation simulation in this paper is done using the Python package QuTiP~\cite{Johansson2012Qutip,Johansson2013Qutip}. 

From the Lindblad master equation simulation, we obtain the actual channel or propagator $\Vactual$ of the $ZZ$ and $XX$ stabilizers which evolve the system from the beginning to the end of each stabilizer's protocol time. We then convert these propagators into their normal forms $\Nactualhat$ using Eq.~\eqref{eq:Ttt}. Afterward, we convert $\Nactualhat$ into the superoperator form and then take the matrix logarithm of $\Nactualhat$ in order to get the time-independent effective Lindbladians $\Lcalonent$ [Eq.~\eqref{eq:Lcaleff}]. We then decompose these effective Lindbladians $\Lcalonent$ using the cluster expansion method in Sec.~\ref{sec:cluster} and exponentiate the decomposed Lindbladians to get the normal-form of the approximate noise channels $\Napproxhat$ as detailed in Sec.~\ref{sec:construction}. To convert the approximate noise channels of the $ZZ$ and $XX$ stabilizers into the standard form  $\Vapproxhat$, we multiply the corresponding $\Napproxhat$ by the ideal (noiseless) target operation $\Ucalhat$, i.e., 
\begin{equation}
\Vapproxhat = \Napproxhat \, \Ucalhat.
\end{equation} 

We note that the effective Lindbladian $\Lcalthat$ [Eq.~\eqref{eq:Lcaleff}] can in general contain higher-order operators than the operators in the Hamiltonian or the dissipative jump operators. This is because the effective Lindbladian $\Lcalthat$ is calculated by taking the logarithm of the noise channel that evolves the system from the beginning to the end of the protocol. As shown in Eq.~\eqref{eq:Vt}, this noise channel is the time-ordered exponentiation of the usual Lindbladian as defined in Eq.~\eqref{eq:Lcal}. The time-ordered exponentiation  can be expressed using the Magnus expansion~\cite{magnus1954exponential,blanes2009magnus} as the exponentiation of infinite series of operators evaluated from nested commutators involving each term in the Hamiltonian and dissipative jump operators. Since some of the Hamiltonian and dissipation terms do not commute, the commutators give rise to higher-body operators in the noise channel which then translates into higher-body operators in the effective Lindbladian $\Lcalthat$. For our example of the 3-qubit processor, the effective Lindbladian contains correlated qubit jump operators, i.e., two- or three-body jump operators, resulting from the commutator of the qubit-qubit couplings in Eq.~\eqref{eq:Hint} and the individual qubit jump operators in Eq.~\eqref{eq:Lindbladappendix}. Moreover, it also comprises three-qubit coupling terms, arising from the commutator between two neighboring two-qubit coupling terms in Eq.~\eqref{eq:Hint}. 

One advantage of the cluster expansion proposed here is that it gives a simple recipe on how to decompose the effective Lindbladian into components based on interqubit correlation degrees without having to resort to the more complicated Magnus expansion technique as described above. Moreover, our cluster expansion has the benefit that it can be used even for the case where we do not know the exact Hamiltonian and/or dissipators in the system but are given the total effective Lindbladian $\Lcalthat$ of the system. This Lindbladian can be obtained experimentally using noise characterization techniques such as the Lindblad tomography~\cite{Samach2022Lindblad} or process tomography~\cite{howard2006quantum} whose data can be fitted to get the effective Lindbladian for the quantum channel~\cite{onorati2021fitting,onorati2023fitting}.  

We use the following procedure to implement the $\dsl2,0,2\dsr$ code (see Algorithm~\ref{algo_202} of Appendix~\ref{sec:algorithmtwo} for the pseudocode):
\begin{enumerate}
\item Choose a Bell state for the input state $|\Psiin\rangle$ of the data qubits. 
\item Reset the syndrome qubit to state $|0\rangle$.
\item Apply a selected error channel, actual $\Vactualhat$ or approximate $\Vapproxhat$ (for different gain factors), of the $ZZ$ stabilizer.
\item Apply a noiseless (perfect) projective measurement of the syndrome qubit in the $Z$ basis.
\item Apply a selected error channel, actual $\Vactualhat$ or approximate $\Vapproxhat$  (for different gain factors), of the $XX$ stabilizer.
\item Apply a noiseless (perfect) projective measurement of the syndrome qubit in the $Z$ basis.
\item Compute the infidelities between the data-qubits density matrices obtained from the noiseless (ideal) $\Videalhat$, actual $\Vactualhat$ and approximate $\Vapproxhat$ channels.
\item Repeat steps 2-7.
\end{enumerate}

We apply the above procedure to implement the $\dsl2,0,2\dsr$ code for two different connectivities: linear and triangle geometries (see Fig.~\ref{fig:geometry}). In order to understand the contribution of noise components with different degrees of correlations, for both geometries we perform simulations using two different approximate noise channels: second-order ($\Napproxhattwo$) and third-order ($\Napproxhatthree$) approximate noise channels [Eq.~\eqref{eq:Napproxm}].

 \begin{figure*}[t!]
\includegraphics[width=\linewidth]{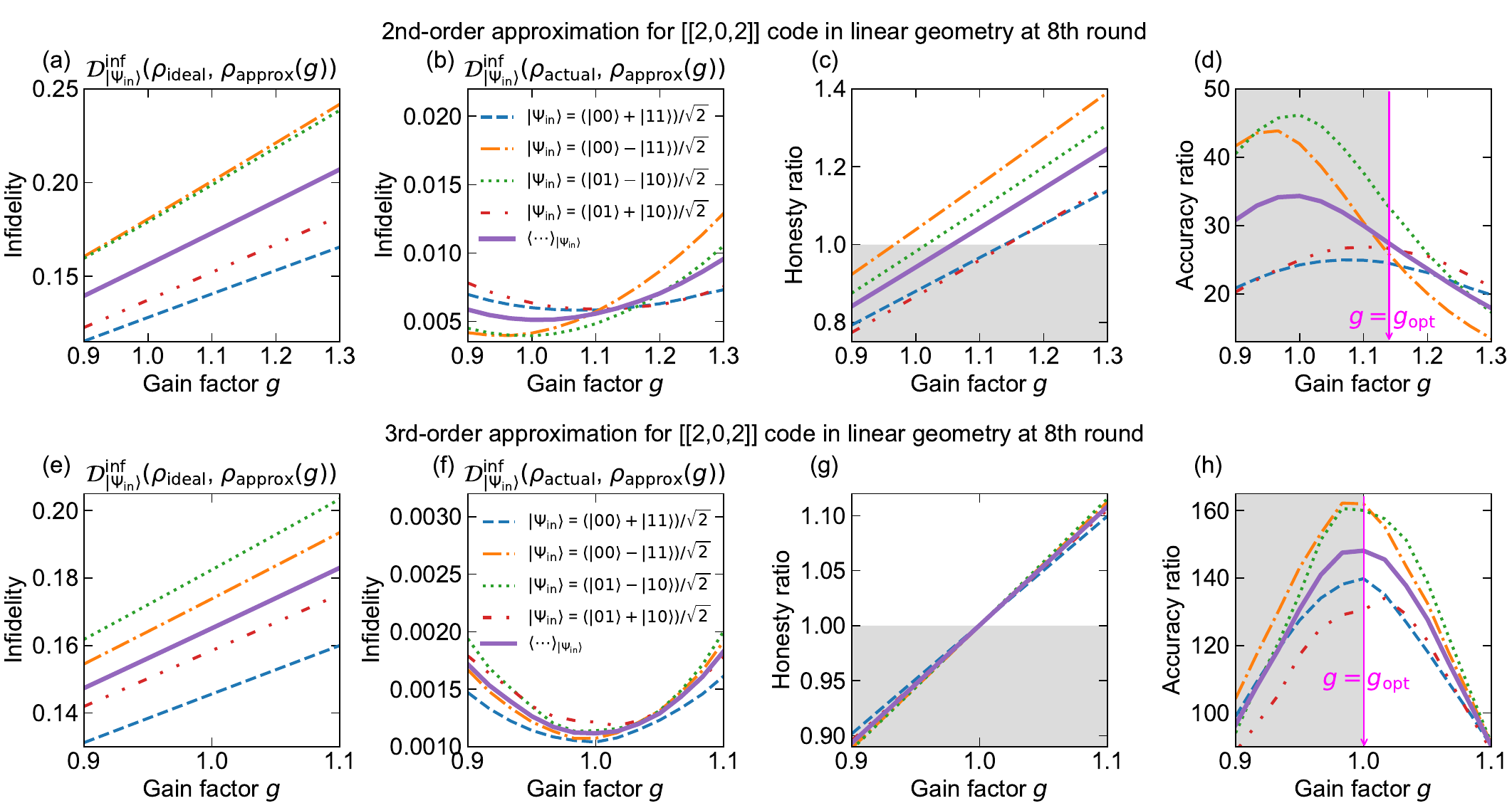}
	\caption{Results simulated for the linear connectivity using the second-order $\Napproxhattwo$ [upper panels: (a), (b), (c), and (d)] and third-order $\Napproxhatthree$ [lower panels: (e), (f), (g), and (h)] approximate noise channels. They are calculated at the 8th round implementation of the $\dsl2,0,2\dsr$ code for different input data-qubits states $|\Psiin\rangle$; the state-averaged results are denoted by $\langle \cdots\rangle_{|\Psiin\rangle}$ and shown as solid purple lines. (a,b,e,f) Infidelity between the data-qubits density matrices as a function of gain factor $g$: (a,e) Infidelity between ideal and approximate density matrices, (b,f) Infidelity between actual and approximate density matrices. (c,g) Honesty ratio as a function of gain factor $g$. (d,h) Accuracy ratio as a function of gain factor $g$. Magenta arrow denotes $g = \gopt$, i.e., the optimal value of $g$ that maximizes the state-averaged accuracy ratio subject to the honesty constraint [Eq.~\eqref{eq:honesty}], where $\gopt\simeq 1.14$  for the second-order and $g = \gopt \simeq 1.0005$ for the third-order approximate channel. Shaded areas in panels (c), (d), (g), and (h) denote regimes where approximate noise models are not honest (honesty ratio $< 1$).}\label{fig:linear_8round}
\end{figure*} 

\subsection{Results and discussion}\label{sec:results}

Figure~\ref{fig:linear_8round}  shows the simulation results for the $\dsl 2,0,2 \dsr$ code, computed using the second-order (upper panels) and third-order (lower panels) approximate noise channels. The results shown are for the 8th round implementation of the code in a linear geometry using different input data qubits states $|\Psiin\rangle$. Figures~\ref{fig:linear_8round}(a) and~\ref{fig:linear_8round}(e) show that the infidelity between the ideal and approximate data-qubits density matrices increases linearly with the gain factor $g$. This is because the gain factor amplifies the noise in the approximate noise channels. Panels (b) and (f) of Fig.~\ref{fig:linear_8round} show that the infidelities between the resulting actual and approximate data-qubits density matrices $\Dactapprox$ exhibit a nonmonotonic dependence on the gain factor $g$, where the minimum is at  $g \simeq 0.95-1.1$ and $g\simeq 1$ for the second-order and third-order approximate noise channels, respectively. 

From these infidelities, we calculate the honesty and accuracy ratios as shown in Figs.~\ref{fig:linear_8round} (c,g) and~\ref{fig:linear_8round}(d,h), respectively. The honesty ratio increases linearly with the gain factor $g$ where the honesty constraint (honesty ratio $\geq 1$) is satisfied for all the four input Bell states  only when $g \gtrsim 1.14$ and $g \gtrsim 1$ for the second-order and third-order approximate channels, respectively. On the other hand, the accuracy ratios have nonmonotonic dependences on the gain factor $g$ with its maximum occurring at $g \simeq 0.95-1.1$ for the second-order approximate noise channel and $g\simeq 1$ for the third-order approximate noise channel. The maximum state-averaged accuracy ratios subject to the honesty constraint are  obtained by choosing the optimal $g = \gopt$ in the honesty regime, where $\gopt\simeq 1.14$ and $\gopt\simeq 1.0005$ for the second- and third-order approximations, respectively. As shown in Fig.~\ref{fig:triangle_8round} of Appendix~\ref{sec:triangle}, the results simulated for the triangle geometry display similar qualitative behaviors as those for the linear geometry.

\begin{figure*}[t!]
\includegraphics[width=\linewidth]{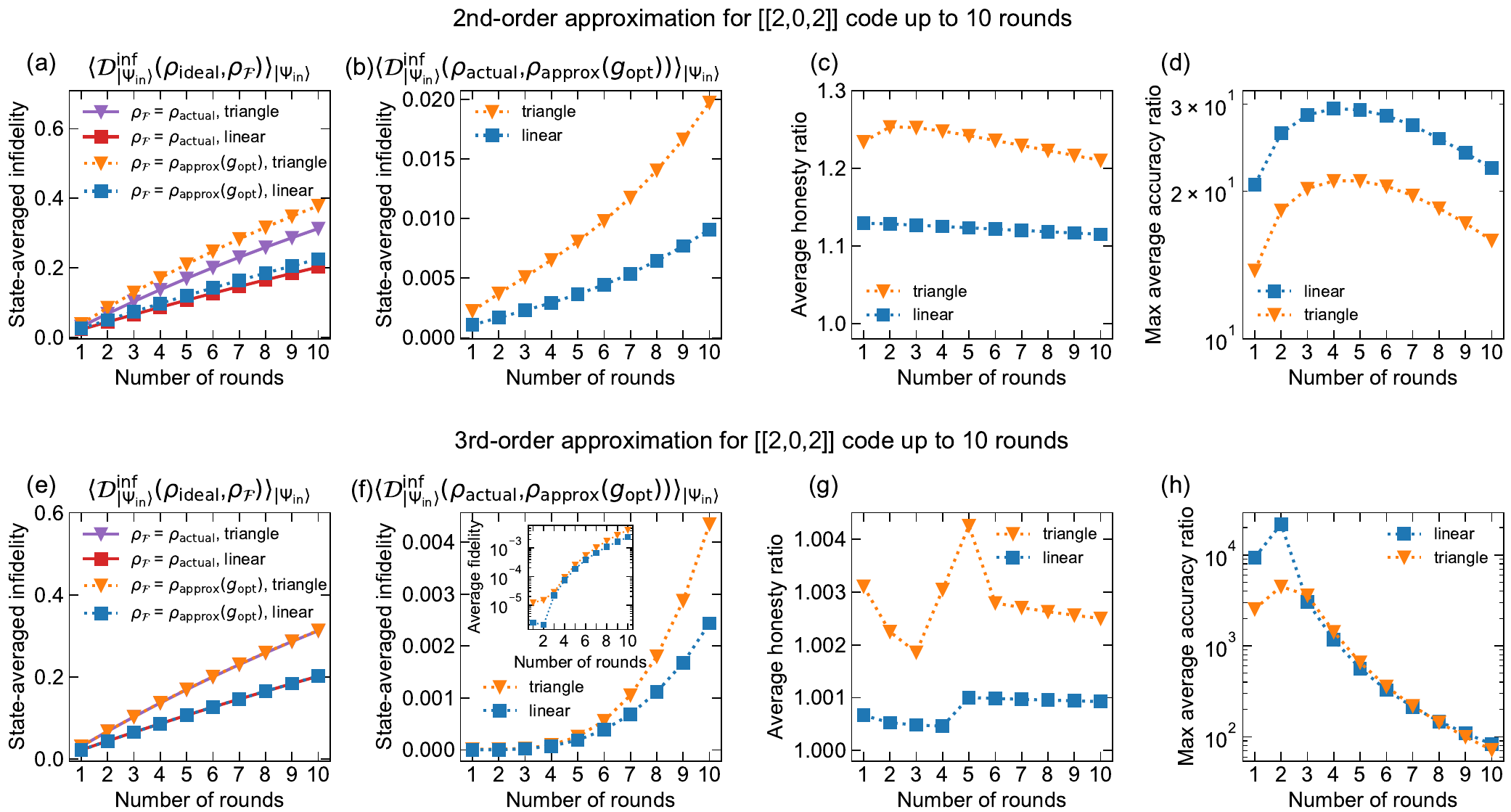}
	\caption{Results simulated by using the second-order $\Napproxhattwo(\gopt)$ [upper panels: (a), (b), (c), and (d)] and third-order $\Napproxhatthree(\gopt)$ [lower panels: (e), (f), (g), and (h)] approximate noise channels, where $\gopt$ is the optimal gain factor chosen to maximize the approximation accuracy subject to the honesty criterion [Eq.~\eqref{eq:honesty}]. The values of $\gopt$ for the second-order approximate noise channels are $\gopt \simeq 1.14$ for linear and $\gopt \simeq 1.2$ for triangle geometry. For the third-order approximation, $\gopt \simeq 1.0005$ for linear and $\gopt \simeq 1.002$ for triangle geometry. Panels (a,b,e,f): State-averaged infidelity between the data-qubits density matrices simulated by using (a,e) ideal and actual noise channels (solid lines),  ideal and approximate noise channels (dotted lines), (b,f) actual and approximate noise channels (Inset of panel f: Semilog plot). Panels (c,d,g,h): (c,g) state-averaged honesty ratio calculated using $\gopt$, (d,h) maximum state-averaged accuracy ratio calculated using $\gopt$.}\label{fig:honest_acc}
\end{figure*} 

Based on the simulations of the infidelities at each round of the $\dsl 2,0,2 \dsr$ code, we plot the state-averaged results for the second-order and third-order approximate noise channels as a function of round number in the upper and lower panels of Fig.~\ref{fig:honest_acc}, respectively. The approximate density matrices $\rhoapprox$ are calculated using the optimal gain factor ($\gopt$) [Eq.~\eqref{eq:Napproxm}], which makes the approximate channels honest and as accurate as possible. Here we show the results for both linear and triangular connectivities. 

Panels (a), (b), (e) and (f) of Fig.~\ref{fig:honest_acc} show the state-averaged infidelities of the data-qubits density matrices calculated after each round of the $\dsl2,0,2\dsr$ code, where the averaging is performed over the four different input Bell states. In Figs.~\ref{fig:honest_acc}(a) and~\ref{fig:honest_acc}(e), the infidelities between the data-qubits density matrices of the ideal ($\rhoideal$) and actual ($\rhoactual$) channels 
 are plotted using solid lines, while the  infidelities between the data-qubits density matrices of the ideal ($\rhoideal$) and approximate ($\rhoapprox$) channels are depicted by dashed lines. Figures~\ref{fig:honest_acc}(b) and ~\ref{fig:honest_acc}(f)  show the infidelities between the data-qubits density matrices of the actual ($\rhoactual$) and approximate ($\rhoapprox$) channels. These infidelities increase with the number of rounds because the noise decoherence effect accumulates as the number of rounds increases, which in turn increases the distance between the density matrices. The increase is larger for the triangle geometry compared to the linear connectivity because the triangle geometry has one more extra crosstalk channel due to the coupling between qubits 1 and 2 (see Fig.~\ref{fig:geometry}). 

The state-averaged honesty ratios for the approximate channels are plotted in panels (c) and (g) of Fig.~\ref{fig:honest_acc}. The results are calculated using the optimal gain factor ($\gopt$) which is the gain factor that maximizes the state-averaged accuracy ratios subject to the constraint that the approximate channel must be honest (i.e., honesty ratio $\geq 1$ [Eq.~\eqref{eq:honesty}]) for all four Bell input states. Using this optimal gain factor $\gopt$, we compute the state-averaged accuracy ratios, which are shown in panels (d) and (h) of Fig.~\ref{fig:honest_acc}. As seen in Figs.~\ref{fig:honest_acc}(d) and~\ref{fig:honest_acc}(h), the accuracy of the approximation first increases with the number of rounds and for a large enough round number ($>$ 5 for the second-order and $>$ 3 for the third-order approximation), the accuracy  decreases with the number of rounds due to the accumulation of the noise decoherence effect. The reason for the nonmonotonic behavior of the accuracy ratio as a function of number of rounds is because the distance between $\rhoideal$ and $\rhoactual$ increases roughly linearly with the number of rounds [Figs.~\ref{fig:honest_acc}(a) and~\ref{fig:honest_acc}(e)] while the distance between $\rhoactual$ and $\rhoapprox$ grows exponentially with the number of rounds [Figs.~\ref{fig:honest_acc}(b) and~\ref{fig:honest_acc}(f)]. Note that in going from the second-order to the third-order approximation, the accuracy ratio increases by $\sim 1-2$ orders of magnitude, which implies that the three-body correlated noise has significant effects on the actual simulation accuracy of the code. 

\section{Conclusions and outlook}\label{sec:conclusions}
We propose a generic perturbative series approach, based on the cluster-expansion method, for decomposing the Lindbladian noise in an actual noise channel into components based on the qubit-qubit correlation degrees. Furthermore, we also present a systematic way to construct honest and accurate approximate noise channels from the decomposed Lindbladians, where the accuracy (efficiency) increases (decreases) as we include correlated noise components with higher and higher correlation degrees. Our work provides a general technique to characterize multi-qubit correlated errors, which goes beyond the Pauli noise model~\cite{harper2020efficient,van2023probabilistic} and also commonly used approximate noise models that assume only one- and two-qubit noise. 

By applying the cluster expansion approach to a concrete physical setting, namely, fixed-frequency superconducting qubits coupled via static always-on interactions, we show that for realistic parameter sets common to this setup, correlated noise beyond two qubits have significant effects on the simulation accuracy of QEC codes. While here we focus on superconducting qubits, our technique can actually be applied to any kind of quantum computing platforms. Moreover, our method can be extended to more complex settings, such as multi-level qubit systems, where effects due to noncomputational states (e.g., leakage) can be taken into account, and also to larger number of qubits.

 We note that while in this paper we apply our method to approximate noise in a simple system of 3-qubit processor by using the input noise data of the whole system, our approach can also be used to construct the approximate noise channels of the whole system by using the input noise data of its smaller subsystems. In particular, our cluster expansion allows us to do this in a way that it would avoid the issue of overcounting some of the Lindbladian components that would appear in the naive approximation where the noise channel of the full system is obtained by simply multiplying together its subsystems' noise channels. For example, suppose we would like to approximate the noise in a 3-qubit processor given the experimental noise characterization or simulation done on its 2-qubit subsystems which supplies the input noise channels: $\Nthat^{\{1,2\}}$, $\Nthat^{\{1,3\}}$, and $\Nthat^{\{2,3\}}$. Since $\Nthat^{\{1,2\}} = \mathrm{exp} (\Lcalthat^{\{1\}} + \Lcalthat^{\{2\}}+\Lcalthat^{\{1,2\}})$, $\Nthat^{\{1,3\}} = \mathrm{exp}(\Lcalthat^{\{1\}} + \Lcalthat^{\{3\}}+\Lcalthat^{\{1,3\}})$, and $\Nthat^{\{2,3\}} = \mathrm{exp}(\Lcalthat^{\{2\}} + \Lcalthat^{\{3\}}+\Lcalthat^{\{2,3\}})$, if one naively approximates the noise of the full system as $\Napproxhattwo = \Nthat^{\{1,2\}}
\cdot\Nthat^{\{1,3\}}\cdot\Nthat^{\{2,3\}}$, the terms $ \mathrm{exp}(g \Lcalthat^{\{1\}})$, $ \mathrm{exp}(g \Lcalthat^{\{2\}})$, and $ \mathrm{exp}(g \Lcalthat^{\{3\}})$ would appear twice in $\Napproxhattwo$ and therefore would not be an accurate representation of the actual noise channel. 

To avoid the overcounting issue, we can first use our cluster expansion approach to extract from the subsystems' noise channels different Lindbladian components that describe various noise terms with different interqubit correlation degrees. Using these Lindbladian components, we then construct an approximate noise channel $\Napproxhat$ for the whole system by selectively ``stitching'' together these Lindbladian terms in a way that avoids the overcounting issue. For the example of approximating noise in a 3-qubit processor using the noise channels of its 2-qubit subsystems, our approach allows us to first extract from $\Napproxhat^{\{1,2\}}$, $\Napproxhat^{\{1,3\}}$, and $\Napproxhat^{\{2,3\}}$ the different Lindbladian components: $\Lcalthat^{\{1\}}$, $\Lcalthat^{\{2\}}$, $\Lcalthat^{\{3\}}$, $\Lcalthat^{\{1,2\}}$, $\Lcalthat^{\{1,3\}}$,  and $\Lcalthat^{\{2,3\}}$. As in Eq.~\eqref{eq:Napproxm}, the approximate noise channel for the 3-qubit system can be then obtained by selectively including the Lindbladian terms in the approximate noise channel such that none of them appears more than once, yielding
\begin{align}
\Napproxhattwo &= \mathrm{exp}(g \Lcalthat^{\{1\}})  \mathrm{exp}(g\Lcalthat^{\{2\}})  \mathrm{exp}(g\Lcalthat^{\{3\}}) \nonumber\\
 &\qquad\mathrm{exp}(g \Lcalthat^{\{1,2\}})  \mathrm{exp}(g\Lcalthat^{\{1,3\}})  \mathrm{exp}(g\Lcalthat^{\{2,3\}}).
\end{align}

We emphasize that the above approach can be extended to approximate the noise of larger systems consisting of $n$ qubits from the input noise channels of its $m$-qubit subsystems. The amount of computational resources needed to construct the approximate noise channels depends on the size of the subsystems ($m$) used in the simulation or experimental noise characterization that provides the input noise data for the cluster expansion.  The smaller the size of the subsystems used in the input noise characterization, the less accurate the constructed approximate noise channel but the more efficient the computation is. For scalability of our method to a large system consisting of $n$ qubits, one therefore needs to choose a small enough subsystem size $m$ which does not grow with the full system size $n$. The input noise data can then be obtained from simulations or experimental noise characterizations done on only at most $n \choose m $ number of subsystems, which require resources that scale only polynomially with $n$ and not exponential in $n$. 
Since in this paper we focus on setting up the formalism of the cluster expansion approach and showing, for clarity, how this method can be applied to a simple system, we leave the task of approximating the noise channel of a larger system from the characterization of its smaller subsystems to future works.

Our technique also has another appealing feature that it gives systematic characterization of the interqubit correlated-noise components order by order. This is because our cluster expansion gives as the output the various Lindbladian terms which one can selectively include or exclude in the construction of the approximate noise channels. The changes of the accuracy of the approximate noise channel as one include or exclude some of the Lindbladian terms can be quantified using the accuracy metric given in our paper or other more scalable metrics. This systematic order-by-order characterization of multi-qubit correlated noise allows us to quantify the smallest maximum degree of noise correlation  or smallest subsystem size in the input noise characterization that needs to be taken into account for a certain desired accuracy in the noise modeling of the full quantum system. Our approach therefore paves the way towards constructing an accurate, honest and scalable approximate simulation of large quantum devices. 

\begin{appendix}
\section{Superoperator representation}\label{sec:superop}
Quantum channels are completely positive and trace preserving (CPTP) linear maps that act on a $\drho \times \drho$ density matrices $\rho = \sum_{i,j=1}^{\drho}\rho_{ij}|e_i\rangle \langle e_j|$. We can represent quantum channels and their generators (i.e., the Lindbladian for our case here) in the superoperator representation where they become $\drho^2 \times \drho^2$ matrices acting on a vectorized $\drho^2 \times 1$ density operator $|\rho\rangle\rangle = \sum_{i,j=1}^{\drho}\rho_{ij}|e_i\rangle \otimes |e_j\rangle$. Operators that act on the
bra and ket now  transform into operators acting on different copies of
the Hilbert space $\mathcal{H}\otimes \mathcal{H}$, i.e., $\hat{A}\rho \hat{B} \rightarrow \hat{A} \otimes \hat{B}^T |\rho \rangle\rangle$, where the superscript $T$ indicates a matrix transpose. As a result, in the superoperator representation, the Lindblad master equation [Eq.\eqref{eq:Lindblad}] is written as $\partial_t|\rho(t)\rangle\rangle = \hat{\mathcal{L}}|\rho(t)\rangle\rangle$, where we use the  symbol $\;\hat{}\;$ to denote the superoperator representation, e.g., $\mathcal{L}\rightarrow\hat{\mathcal{L}}$.

Using the superoperator representation, we write the Lindbladian as
\begin{align}
\hat{\mathcal{L}} &= -i(\hat{H} \otimes \mathds{1} - \mathds{1}\otimes \hat{H}^T) \nonumber\\
&\hspace{0.5 cm}+\sum_{j}\gamma_j\left(\hat{A}_j \otimes \hat{A}_j^* -\frac{1}{2} \hat{A}_j^\dagger \hat{A}_j \otimes \mathds{1} -\frac{1}{2} \mathds{1} \otimes \hat{A}_j^T \hat{A}_j^* \right).
\end{align}
As a result, the action of a superoperator on a density matrix becomes simply a matrix-vector multiplication. Furthermore, the composition of superoperators is given by a product of the matrix representations of the superoperators. In the superoperator representation, a quantum channel is therefore given by the matrix exponentiation of its Lindbladian generator, i.e., $\hat{\Vt} = \hat{T}\left[\exp(\int_0^t\Lcalhat(t') dt')\right]$. At any given time $t$, the state of the system is then given by $|\rho(t)\rangle\rangle = \hat{\Vt}|\rho(0)\rangle\rangle$.

\section{Proof of the cluster expansion}\label{sec:proofrecur}
\begin{theorem}\label{theorem_one}
Let $\Lcalonent$ be the Lindbladian of a system comprising $n$ qubits where each of the qubits live in the Hilbert space $SU(d)$ and $\Lcalonent^{S_m} = \left( \frac{1}{d^{|\bar{S}_m|}}\mathrm{Tr}_{\bar{S}_m} \left[\Lcalonent \right]\otimes \mathds{1}_{\bar{S}_m} \right) - \sum_{R\subsetneq S_m} \Lcalonent^{R}$ be the Lindbladian term that describes a pure $m$-th body correlation among qubits in the subset $S_m$ where $|S_m| = m$ with $m \leq n$, then $\Lcalonent$ can be decomposed as $\Lcalonent = \sum_{m}^n \sum_{S_m} \Lcalonent^{S_m}$.  
\end{theorem}
\begin{proof}
\begin{align}\label{eq:proofcompleteness}
\Lcalonent &= \sum_{m}^n \sum_{S_m} \Lcalonent^{S_m}\nonumber\\
&= \Lcalonent^{S_n} + \sum_{m}^{n-1} \sum_{S_m} \Lcalonent^{S_m}\nonumber\\
&= \mathrm{Tr}_{\emptyset}\left( \Lcalonent\right) - \sum_{m}^{n-1} \sum_{S_m} \Lcalonent^{S_m} + \sum_{m}^{n-1} \sum_{S_m} \Lcalonent^{S_m}\nonumber\\
&= \Lcalonent,
\end{align}
where in the third line, we have substituted in Eq.~\eqref{eq:Lsmdef} for $\Lcalonent^{S_n}$. 
\end{proof}

\begin{theorem}\label{theorem_two}
Let $\Lcalonent$ be the Lindbladian of a system comprising $n$ qubits where each of the qubits live in the Hilbert space $SU(d)$ and $\Lcalonent^{S_m} = \left( \frac{1}{d^{|\bar{S}_m|}}\mathrm{Tr}_{\bar{S}_m} \left[\Lcalonent \right]\otimes \mathds{1}_{\bar{S}_m} \right) - \sum_{R\subsetneq S_m} \Lcalonent^{R}$ be the component of the Lindbladian that describes the $m$-th order correlation among the qubits in the subset $S_m$, then $\Lcalonent^{S_m}$ can be expressed in terms of identities on all qubits and generators that have non-trivial basis only on all qubits $j \in S_m$, i.e., 
\begin{align}
  \Lcalonent^{S_m} &= c_{m}|w_{\{0,0,\cdots,0\}}|^2 (\mathds{1} \otimes \mathds{1}) \nonumber\\
	&+ \sum_{\substack{\{\alpha_{j}\}, \{\alpha'_j\}\\ \alpha_j+\alpha'_j \neq 0 \\ \mathrm{for}\, j \in S_m,\\ \alpha_j=\alpha'_j = 0 \\ \mathrm{for}\, j \notin S_m}}  w_{\{\alpha_j\}} w_{\{\alpha'_j\}}^*
\left(\bigotimes_j B^{\alpha_j}_j\right)  \otimes \left(\bigotimes_j\left[B_j^{\alpha_j'}\right]^*\right),
\end{align}
where $c_m$, $w_{\{\alpha_j \}}$ are constants, $B^{0}_j = \mathds{1}$ on qubit $j$, and $B^{\alpha_j}_j$ for $\alpha_j \neq 0$ are non-trivial basis generators $B^{\alpha_j}$ acting on qubit $j$, satisfying the relations $\mathrm{Tr} [B^{\alpha}] = 0$ and $\mathrm{Tr}(B^\alpha B^\beta) = \delta_{\alpha\beta}$.
\end{theorem}
\begin{proof}
We begin by noting that any complex superoperator $\hat{O}$ in the double Hilbert space $SU(d) \otimes SU(d)$ of a qubit with $d$ number of energy levels can be expressed as
\begin{equation}\label{eq:matrix_decomposition}
\hat{O} = \sum_{\alpha,\alpha' = 0,1,\cdots,d^2-1}w_\alpha w_{\alpha'}^* B^\alpha \otimes B^{\alpha'},  
\end{equation}
where $B^0 = \mathds{1}$ and $B^\alpha$ for $\alpha \neq 0$ are the generators of the Lie algebra satisfying the commutation relation 
\begin{equation}
[B^{\alpha}, B^{\beta}] = 2i \sum_{\gamma}f_{\alpha\beta\gamma} B^\gamma,
\end{equation}
and $w_{\alpha,\alpha'}$ are the coefficients of the different basis generators.
Eq.~\eqref{eq:matrix_decomposition} implies the completeness of the basis generators.
Here $f_{\alpha\beta\gamma}$ are the structure constants which are completely antisymmetric with respect to the interchange of pair of indices, generalizing the Levi-Civita symbol $\epsilon_{\alpha\beta\gamma}$ of $\sufrak(2)$. Specifically, the generators $B^\alpha$ and $B^\beta$ for $\alpha, \beta \neq 0$  have the property
\begin{equation}\label{eq:traceless}
\mathrm{Tr} B^\alpha = 0,
\end{equation}
and satisfy the normalization condition:
\begin{equation}
\mathrm{Tr}(B^\alpha B^\beta) = \delta_{\alpha\beta},
\end{equation}
where $\delta_{\alpha\beta}$ is the Kronecker delta function.
As an example, for $\sufrak(2)$ Lie algebra, $B^\alpha = \sigma^\alpha$ which are the Pauli matrices and for $\sufrak(3)$, $B^\alpha = \lambda^\alpha$, the Gell-Mann matrices. 

A Lindbladian superoperator in the double Hilbert space $[SU(d)]^{\otimes n} \otimes [SU(d)]^{\otimes n}$ of $n$ qubits can therefore be written as
\begin{align}\label{eq:lindbladdecompose}
\Lcalthat &= \sum_{\{\alpha_j\},\{\alpha'_j\}} w_{\{\alpha_j\}} w^*_{\{\alpha'_j\}} \left( \bigotimes_j B_j^{\alpha_j}  \right) \otimes \left(\bigotimes_j \left[B_j^{\alpha'_j}\right]^* \right),
\end{align}
where $B^{\alpha_j}_j$ is the basis generator $B^{\alpha_j}$ acting on qubit $j$. To obtain the Lindbladian component $\Lcalthat^{S_m}$ that describes a pure $m$-th order noise correlation, we apply the recurrence relation in Eq.~\eqref{eq:Lsmdef}. The first term in Eq.~\eqref{eq:Lsmdef} involves taking a partial trace of the full Lindbladian $\Lcalthat$ [Eq.~\eqref{eq:lindbladdecompose}] over qubits outside the subset $S_m$. This gives a term
\begin{align}\label{eq:firstterm_of_recur}
&\frac{1}{d^{|\bar{S}_m|}}\mathrm{Tr}_{\bar{S}_m} \left[\Lcalonent \right]\otimes \mathds{1}_{\bar{S}_m}  \nonumber\\
&= \sum_{\substack{\{\alpha_j\},\{\alpha'_j\}\\ \alpha_j = \alpha'_j = 0 \\ \mathrm{for\,} j \notin S_m}} w_{\{\alpha_j\}} w^*_{\{\alpha'_j\}} \left( \bigotimes_j B_j^{\alpha_j}  \right) \otimes \left(\bigotimes_j \left[B_j^{\alpha'_j}\right]^* \right),
\end{align} 
which can have non-trivial basis supports only on qubit $j \in S_m$. This means that after the partial trace, we are left only with the terms $\Lcalthat^{R}$ for all $R \subseteq S_m$ as all the terms $\Lcalthat^{\bar{R}}$, where $\bar{R}$ is the complement of $R$, drop out as these terms have at least one non-trivial basis on a qubit $j \in \bar{S}_m$ and the trace of this non-trivial basis is zero [Eq.~\eqref{eq:traceless}]. In Eq.~\eqref{eq:firstterm_of_recur}, we divide the result after the partial trace by $d^{|\bar{S}_m|}$ to cancel the same constant obtained from tracing over the identity matrices on qubits in $\bar{S}_m$. 

The reduced Lindbladian after the partial trace can be expressed in terms of the basis generators in the Hilbert space of the qubits in the subset $S_m$ where these basis generators span a complete basis set in this Hilbert subspace. Therefore, the partial trace gives  Lindbladian components that can describe any possible noise correlation (from the first-order to $m$-th order correlation) in the subsystem $S_m$ consisting of $m$ qubits. To obtain the Lindbladian component $\Lcalthat^{S_m}$ that describes a pure $m$-th order noise correlation, we need to subtract all the Lindbladians describing the lower [from the first to the $(m-1)^\mathrm{th}$]-order noise correlations from the partial trace result as shown by the second term in Eq.~\eqref{eq:Lsmdef}. This observation allows us to use a recurrence relation as in Eq.~\eqref{eq:Lsmdef} to systematically extract the Lindbladian components order by order starting from the first order. 

The algorithm to find each of the terms in the cluster expansion is as follows. First, we calculate all the first-order Lindbladian terms $\Lcalthat^{S_1}$  by tracing the full Lindbladian $\Lcalthat$ over all qubits except the qubit of interest. Each of the resulting first-order terms is a component in the full Lindbladian that has non-trivial basis supports only on a single qubit, and therefore describes the first-order noise in each individual qubit. Next, we do the partial trace of the full Lindbladian over all qubits except a pair of qubit. This gives a reduced Lindbladian that span a complete basis in the Hilbert space of a pair of qubit. To get the second-order $\Lcalthat^{S_2}$ Lindbladian term which describes a pure second-order noise correlation between the qubits, we have to subtract the first-order terms (obtained from the first step) from the reduced Lindbladian after the partial trace. Doing this for all pairs of qubits gives all the second-order terms. We can repeat this procedure of doing partial trace and then subtracting the lower-order terms from the partial-trace result to obtain the higher- and higher-order terms. The algorithm for the recurrence relation can be found in Algorithm~\ref{algo_recur}. The end result is that $\Lcalonent^{S_m}$ can be expressed using identities on all qubits and generators with non-trivial basis on all of the $m$ qubits in the subset $S_m$, i.e.,
\begin{align}\label{eq:LSm}
  &\Lcalonent^{S_m} \nonumber\\
	&=\left( \frac{1}{d^{|\bar{S}_m|}}\mathrm{Tr}_{\bar{S}_m} \left[\Lcalonent \right]\otimes \mathds{1}_{\bar{S}_m} \right) - \sum_{R\subsetneq S_m} \Lcalonent^{R} \nonumber\\
		&=(-1)^{m+1}|w _{\{0,0,\cdots,0\}}|^2 (\mathds{1} \otimes \mathds{1}) \nonumber\\\nonumber\\
	&+ \sum_{\substack{\{\alpha_{j}\}, \{\alpha'_j\}\\ \alpha_j+\alpha'_j \neq 0 \\ \mathrm{for}\, j \in S_m,\\
	\alpha_j = \alpha'_j = 0 \\ \mathrm{for}\, j \notin S_m
	} }  w_{\{\alpha_j\}} w_{\{\alpha'_j\}}^*
\left(\bigotimes_{j} B^{\alpha_j}_j\right)  \otimes \left(\bigotimes_{j}\left[B_j^{\alpha_j'}\right]^*\right).
\end{align}
The condition $\alpha_j + \alpha'_j \neq 0$ in the sum of Eq.~\eqref{eq:LSm} is to ensure that the operator has nontrivial basis generators on qubit $j$ in the subset $S_m$ and the condition $\alpha_j = \alpha'_j = 0$ is for the operator to have only the trivial identity basis on  qubit $j \notin S_m$. 
\end{proof}

As a simple example, let us consider the case of a quantum processor consisting of two 2-level qubits. We can use Theorem~\ref{theorem_one} to express the total Lindbladian as
\begin{equation}
\Lcalthat = \Lcalonent^{\{1\}} + \Lcalonent^{\{2\}}+ \Lcalonent^{\{1,2\}}.
\end{equation}
Using Theorem~\ref{theorem_two} and Eq.~\eqref{eq:Lsmdef}, we can write the above decomposed Lindbladian components in terms of their basis generators as
\begin{widetext}
\begin{align}
\Lcalonent^{\{1\}} &= |w_{\{0,0\}}|^2 (\mathds{1}_{\{1,2 \}} \otimes \mathds{1}_{\{1,2\}}) + \sum_{\alpha_1=1,2,3} w_{\{\alpha_1,0\}} w_{\{0,0\}}^* \left( \left(\sigma^{\alpha_1}_1 \otimes \mathds{1}_{\{2\}} \right) \otimes \left( \mathds{1}_{\{1\}} \otimes \mathds{1}_{\{2\}}\right)\right) &\nonumber\\
&\quad+  \sum_{\alpha'_1=1,2,3} w_{\{0,0\}} w_{\{\alpha'_1,0\}}^* \left( \left( \mathds{1}_{\{1\}} \otimes \mathds{1}_{\{2\}}\right) \otimes \left((\sigma^{\alpha'_1}_1)^* \otimes  \mathds{1}_{\{2\}} \right)\right) \nonumber\\
&\quad+ \sum_{\alpha_1, \alpha'_1=1,2,3} w_{\{\alpha_1,0\}} w_{\{\alpha'_1,0\}}^* \left( \left( \sigma^{\alpha_1}_1 \otimes \mathds{1}_{\{2\}}\right) \otimes \left((\sigma^{\alpha'_1}_1)^* \otimes  \mathds{1}_{\{2\}} \right)\right), \nonumber\\
\Lcalonent^{\{2\}} &= |w_{\{0,0\}}|^2 (\mathds{1}_{\{1,2 \}} \otimes \mathds{1}_{\{1,2\}}) + \sum_{\alpha_2=1,2,3} w_{\{0,\alpha_2\}} w_{\{0,0\}}^* \left( \left( \mathds{1}_{\{1\}} \otimes \sigma^{\alpha_2}_2 \right) \otimes \left( \mathds{1}_{\{1\}} \otimes \mathds{1}_{\{2\}}\right)\right) &\nonumber\\
&\quad+  \sum_{\alpha'_2=1,2,3} w_{\{0,0\}} w_{\{0,\alpha'_2\}}^* \left( \left( \mathds{1}_{\{1\}} \otimes \mathds{1}_{\{2\}}\right) \otimes \left(  \mathds{1}_{\{2\}} \otimes (\sigma^{\alpha'_2}_2)^* \right)\right)\nonumber\\
&\quad + \sum_{\alpha_2,\alpha'_2=1,2,3}  w_{\{0,\alpha_2\}} w_{\{0,\alpha'_2\}}^* \left( \left(  \mathds{1}_{\{1\}} \otimes \sigma^{\alpha_2}_2  \right) \otimes \left( \mathds{1}_{\{1\}}  \otimes (\sigma^{\alpha'_2}_2)^*\right)\right), \nonumber\\
\Lcalonent^{\{1,2\}} &= -|w_{\{0,0\}}|^2 (\mathds{1}_{\{1,2 \}} \otimes \mathds{1}_{\{1,2\}}) + \sum_{\alpha_1, \alpha_2=1,2,3}  w_{\{\alpha_1,\alpha_2\}} w_{\{0,0\}}^* \left( \left( \sigma^{\alpha_1}_1  \otimes \sigma^{\alpha_2}_2 \right) \otimes \left( \mathds{1}_{\{1\}} \otimes \mathds{1}_{\{2\}}\right)\right) &\nonumber\\
&\quad+  \sum_{\alpha'_1,\alpha'_2=1,2,3} w_{\{0,0\}} w_{\{\alpha'_{1},\alpha'_2\}}^* \left( \left( \mathds{1}_{\{1\}} \otimes \mathds{1}_{\{2\}}\right) \otimes \left(  (\sigma^{\alpha'_1}_1)^* \otimes (\sigma^{\alpha'_2}_2)^* \right)\right) \nonumber\\
&\quad+  \sum_{\alpha_1,\alpha'_2=1,2,3} w_{\{\alpha_1,0\}} w_{\{0,\alpha'_2\}}^* \left( \left(\sigma^{\alpha_1}_1\otimes \mathds{1}_{\{2\}}\right) \otimes \left(  \mathds{1}_{\{1\}} \otimes (\sigma^{\alpha'_2}_2)^* \right)\right) \nonumber\\
&\quad+  \sum_{\alpha_2,\alpha'_1=1,2,3} w_{\{0,\alpha_2\}} w_{\{\alpha'_1,0\}}^* \left( \left( \mathds{1}_{\{1\}} \otimes \sigma^{\alpha_2}_2\right) \otimes \left( (\sigma^{\alpha'_1}_1)^* \otimes \mathds{1}_{\{2\}}  \right)\right) \nonumber\\
&\quad+  \sum_{\alpha_1,\alpha_2,\alpha'_1=1,2,3} w_{\{\alpha_1,\alpha_2\}} w_{\{\alpha'_1,0\}}^* \left( \left( \sigma^{\alpha_1}_1 \otimes \sigma^{\alpha_2}_2\right) \otimes \left( (\sigma^{\alpha'_1}_1)^* \otimes \mathds{1}_{\{2\}}  \right)\right) \nonumber\\
&\quad+  \sum_{\alpha_1,\alpha_2,\alpha'_2=1,2,3} w_{\{\alpha_1,\alpha_2\}} w_{\{0,\alpha'_2\}}^* \left( \left( \sigma^{\alpha_1}_1 \otimes \sigma^{\alpha_2}_2\right) \otimes \left( \mathds{1}_{\{1\}} \otimes (\sigma^{\alpha'_2}_2)^*   \right)\right) \nonumber\\
&\quad+  \sum_{\alpha_1,\alpha'_1,\alpha'_2=1,2,3} w_{\{\alpha_1,0\}} w_{\{\alpha'_1,\alpha'_2\}}^* \left( \left( \sigma^{\alpha_1}_1 \otimes \mathds{1}_{\{2\}}\right) \otimes \left( (\sigma^{\alpha'_1}_1)^*\otimes (\sigma^{\alpha'_2}_2)^*   \right)\right) \nonumber\\
&\quad+  \sum_{\alpha_2,\alpha'_1,\alpha'_2=1,2,3} w_{\{0,\alpha_2\}} w_{\{\alpha'_1,\alpha'_2\}}^* \left( \left( \mathds{1}_{\{1\}}\otimes \sigma^{\alpha_2}_2\right) \otimes \left( (\sigma^{\alpha'_1}_1)^*\otimes (\sigma^{\alpha'_2}_2)^*   \right)\right) \nonumber\\
&\quad+ \sum_{\alpha_1,\alpha_2, \alpha'_1, \alpha'_2=1,2,3} w_{\{\alpha_1,\alpha_2\}} w_{\{\alpha'_1,\alpha'_2\}}^* \left( \left( \sigma^{\alpha_1}_1   \otimes \sigma^{\alpha_2}_2  \right) \otimes \left(  (\sigma^{\alpha'_1}_1)^* \otimes (\sigma^{\alpha'_2}_2)^*\right)\right),
\end{align}
\end{widetext}
with $\sigma^{1,2,3}_j$ being the Pauli matrices acting on qubit $j$.

Now, we are going to derive Eq.~\eqref{eq:recur} from Eq.~\eqref{eq:Lsmdef}. To this end, we will use the following two lemmas (Lemma~\ref{lemma_one} and~\ref{lemma_two}):
\begin{lemma}\label{lemma_one}
Let $A$ be a matrix such that $||A|| \leq \varepsilon < 1$, then $e^{A} = \mathds{1} + A + R$ where $||R|| \leq e^{\varepsilon} - (1 + \varepsilon)$.
\end{lemma}
\begin{proof}
Since $e^A \equiv 1+ A+ A^2/2! + A^3/3! + A^4/4! +\cdots$, then $R = A^2/2! + A^3/3! + A^4/4!+ \cdots$. As a result, we have
\begin{align}
||R|| &\leq \frac{||A||^2}{2!} + \frac{||A||^3}{3!}  + \frac{||A||^4}{4!}  + \cdots, \nonumber\\
&\leq \frac{\varepsilon^2}{2!} + \frac{\varepsilon^3}{3!} + \frac{\varepsilon^4}{4!} + \cdots \nonumber\\
&\leq e^{\varepsilon} - (1 + \varepsilon).
\end{align}
The above lemma is proven as Lemma 12 in Ref.~\cite{krovi2023improved}.
\end{proof}

\begin{lemma}\label{lemma_two}
Let $A$ be a matrix such that $||A|| \leq \varepsilon < 1$, then $\mathrm{ln}(\mathds{1} + A) =  A + R $ where $||R|| \leq  \mathrm{ln} \left[e^{-\varepsilon}/(1-\varepsilon) \right]$.
\end{lemma}
\begin{proof}
Since $\mathrm{ln}(\mathds{1} + A) \equiv A - A^2/2 + A^3/3 - A^4/4+ \cdots$, then $R = -A^2/2 + A^3/3 - A^4/4+ \cdots$ .
\begin{align}
||R|| &\leq \frac{||A||^2}{2} + \frac{||A||^3}{3} + \frac{||A||^4}{4} + \cdots, \nonumber\\
&\leq \frac{\varepsilon^2}{2} + \frac{\varepsilon^3}{3} + \frac{\varepsilon^4}{4} + \cdots \nonumber\\
& \leq  \mathrm{ln} \left(\frac{e^{-\varepsilon}}{1 - \varepsilon}\right).
\end{align}
where in going to the third line, we have used the series expansion $\mathrm{ln} \left[e^{-\varepsilon}/(1-\varepsilon)\right] = \sum_{n = 2} \varepsilon^{n}/n$.
\end{proof}

Using Lemma~\ref{lemma_one} and~\ref{lemma_two}, we now derive Eq.~\eqref{eq:recur} from Eq.~\eqref{eq:Lsmdef}.
\begin{theorem}\label{theorem_three}
Let $||\Lcalthat|| \leq \varepsilon < \mathrm{ln}(2)$, then 
\begin{align}
&\left|\left|\mathrm{ln}\left[\frac{1}{d^{|\bar{S}_m|}}  \mathrm{Tr}_{\bar{S}_m}\left( e^{\Lcalthat} \right)  \right] - \frac{1}{d^{|\bar{S}_m|}} \mathrm{Tr}_{\bar{S}_m} \left(\Lcalthat \right)  \right|\right|\nonumber\\
% & <\frac{ (\varepsilon + e\varepsilon^2/2)^2 }{2 \left(1 - \varepsilon -e\varepsilon^2/2\right)} + e\varepsilon^2/2
&\leq\mathrm{ln}  \left(\frac{\mathrm{exp}(1- e^{\varepsilon}) }{2 - e^{\varepsilon}} \right) + e^{\varepsilon} - 1 - \varepsilon.
\end{align}
\end{theorem}
\begin{proof}
We first prove the following inequality 
\begin{align}
&\left|\left|\mathrm{ln}\left[\frac{1}{d^{|\bar{S}_m|}}\mathrm{Tr}_{\bar{S}_m}\left(\mathds{1} + \Lcalthat + \mathcal{R} \right)\right] - \frac{1}{d^{|\bar{S}_m|}} \mathrm{Tr}_{\bar{S}_m} \left(\Lcalthat + \mathcal{R}  \right)  \right|\right| \nonumber\\
&\leq \mathrm{ln} \left[\frac{ \mathrm{exp}\left(-\left|\left| \frac{1}{d^{|\bar{S}_m|}} \mathrm{Tr}_{\bar{S}_m} \left[ \Lcalthat + \mathcal{R}  \right]  \right|\right|\right)}{\left(1 -\left|\left| \frac{1}{d^{|\bar{S}_m|}} \mathrm{Tr}_{\bar{S}_m} \left[ \Lcalthat + \mathcal{R}  \right]  \right|\right| \right)}\right] \nonumber\\ 
&\leq \mathrm{ln}\left[ \frac{( \mathrm{exp} \left(-|| \Lcalthat || - || \mathcal{R}  ||\right)}{\left(1 -|| \Lcalthat || - || \mathcal{R}  || \right)}\right]\nonumber\\
&\leq\mathrm{ln}  \left(\frac{\mathrm{exp}(1- e^{\varepsilon}) }{2 - e^{\varepsilon}} \right),
\end{align}
where in going to the second line we have used Lemma~\ref{lemma_two} and in going to the last line, we have used Lemma~\ref{lemma_one}.
Using the triangle inequality, we then have
\begin{align}\label{eq:errorapprox}
&\left|\left|\mathrm{ln}\left[\frac{1}{d^{|\bar{S}_m|}}  \mathrm{Tr}_{\bar{S}_m}\left( e^{\Lcalthat} \right)  \right] - \frac{1}{d^{|\bar{S}_m|}} \mathrm{Tr}_{\bar{S}_m} \left(\Lcalthat \right)  \right|\right| \nonumber\\
&\leq\left|\left|\mathrm{ln}\left[\frac{1}{d^{|\bar{S}_m|}}\mathrm{Tr}_{\bar{S}_m}\left(\mathds{1} + \Lcalthat + \mathcal{R} \right)\right] - \frac{1}{d^{|\bar{S}_m|}} \mathrm{Tr}_{\bar{S}_m} \left(\Lcalthat + \mathcal{R}\right)   \right|\right| \nonumber\\
&\quad + \left| \left|\frac{1}{d^{|\bar{S}_m|}}\left(\mathrm{Tr}_{\bar{S}_m} \mathcal{R}\right) \right|\right| \nonumber\\
&\leq\mathrm{ln}  \left(\frac{\mathrm{exp}(1- e^{\varepsilon}) }{2 - e^{\varepsilon}} \right) +e^{\varepsilon} - (1 + \varepsilon).
\end{align}
\end{proof}
Furthermore, by using Theorem~\ref{theorem_three} we have for $\Lcalonent^{S_m} \equiv\mathrm{Tr}_{\bar{S}_m} \left(\Lcalonent \right)\otimes \frac{\mathds{1}_{\bar{S}_m}}{d^{|\bar{S}_m|}}  - \sum_{R\subsetneq S_m} \Lcalonent^{R}$ in the small noise regime $||\Lcalthat|| \leq \varepsilon < 1/2$,, the following bound for the error in the approximation of Eq.~\eqref{eq:recur} in the main text:
\begin{align}\label{eq:Lsm_ineq}
&\left|\left|\Lcalonent^{S_m} - \left( \mathrm{\ln} \left[\mathrm{Tr}_{\bar{S}_m} \left(e^{\Lcalonent} \right)\otimes \frac{\mathds{1}_{\bar{S}_m}}{d^{|\bar{S}_m|}}\right]  - \sum_{R\subsetneq S_m} \Lcalonent^{R} \right)\right|\right|\nonumber\\
&\qquad\qquad\leq\mathrm{ln} \left(\frac{\mathrm{exp}(1- e^{\varepsilon}) }{2 - e^{\varepsilon}} \right) +e^{\varepsilon} - (1 + \varepsilon)\nonumber\\
&\qquad\qquad<\frac{ (\varepsilon + e\varepsilon^2/2)^2 }{2 \left(1 - \varepsilon -e\varepsilon^2/2\right)} + e\varepsilon^2/2.
\end{align}
In the last line of Eq.~\eqref{eq:Lsm_ineq}, we have used the following relations:
\begin{align}\label{eq:expvarepsilon}
e^{\varepsilon} - (1 + \varepsilon)&= \frac{\varepsilon^2}{2!} + \frac{\varepsilon^3}{3!} + \frac{\varepsilon^4}{4!} + \cdots \nonumber\\
& = \frac{\varepsilon^2}{2} \left(1+\frac{2\varepsilon}{3!} + \frac{2\varepsilon^2}{4!} + \cdots  \right)\nonumber\\
 &< \frac{\varepsilon^2}{2} \left(1+\frac{1}{2!} + \frac{1}{3!} + \cdots  \right)\nonumber\\
&< e\frac{\varepsilon^2}{2},
\end{align}
and
\begin{align}
\mathrm{ln} \left(\frac{e^{-\delta} }{1 -\delta} \right)  & = \frac{\delta^2}{2} \left(1+\frac{2\delta}{3} + \frac{2\delta^2}{4} + \cdots  \right)\nonumber\\
 &< \frac{\delta^2}{2} \left(1+\delta+ \delta^2 + \cdots  \right)\nonumber\\
 &< \frac{\delta^2}{2(1-\delta)},
\end{align}
where
\begin{align}
\delta &= e^\varepsilon -1 \nonumber\\
&< \varepsilon + e\varepsilon^2/2,
\end{align}
by Eq.~\eqref{eq:expvarepsilon}.

\section{Example applications of the cluster expansion approach}\label{app:examples}
In this Appendix, we give two examples of the application of our cluster expansion method. As the first example, we consider the simplest case, i.e., a Lindbladian of a three-qubit quantum processor which consists of only a jump operator acting on qubit 1:
\begin{equation}
\Lcalonethreet = \Lambdaonehat,
\end{equation}
where $\Lambdaone (\rho) =\gamma_1\left(\Aonehat \rho \Aonehatd -\frac{1}{2}\Aonehatd\Aonehat\rho - \frac{1}{2}\rho\Aonehatd\Aonehat \right) $ contains only the jump operator $\Aonehat$ acting on qubit 1. Using the cluster expansion in Eq.~\eqref{eq:recur}, we have the decomposed Lindbladian as\begin{subequations}\label{eq:exampleone}
\begin{align}
\Lcalonethreet^{\{1\}} &= \mathrm{ln}\left( \mathrm{Tr}_{23} \left[e^{\Lambdaonehat }\left( \mathds{1}_{\{1\}}\otimes\frac{\mathds{1}_{\{2,3\}}}{4} \right)\right]\otimes \mathds{1}_{\{2,3\}} \right) \nonumber\\
& = \Lambdaonehat, \\
\Lcalonethreet^{\{2\}} &= \mathrm{ln}\left( \mathrm{Tr}_{13} \left[e^{\Lambdaonehat }\left( \mathds{1}_{\{2\}}\otimes\frac{\mathds{1}_{\{1,3\}}}{4} \right)\right]\otimes \mathds{1}_{\{1,3\}} \right) \nonumber\\
& = 0,\\
\Lcalonethreet^{\{3\}} &= \mathrm{ln}\left( \mathrm{Tr}_{12} \left[e^{\Lambdaonehat}\left( \mathds{1}_{\{3\}}\otimes\frac{\mathds{1}_{\{1,2\}}}{4} \right)\right]\otimes \mathds{1}_{\{1,2\}} \right) \nonumber\\
& = 0,\\
\Lcalonethreet^{\{1,2\}} &= \mathrm{ln}\left( \mathrm{Tr}_{3} \left[e^{\Lambdaonehat }\left( \mathds{1}_{\{1,2\}}\otimes\frac{\mathds{1}_{\{3\}}}{2} \right)\right]\otimes \mathds{1}_{\{3\}} \right) \nonumber\\
&\hspace{0.5cm}- \Lcalonethreet^{\{1\}} - \Lcalonethreet^{\{2\}}   \nonumber\\& = 0,\\
\Lcalonethreet^{\{1,3\}} &= \mathrm{ln}\left( \mathrm{Tr}_{2} \left[e^{\Lambdaonehat }\left( \mathds{1}_{\{1,3\}}\otimes\frac{\mathds{1}_{\{2\}}}{2} \right)\right]\otimes \mathds{1}_{\{2\}} \right) \nonumber\\
&\hspace{0.5cm} - \Lcalonethreet^{\{1\}}   - \Lcalonethreet^{\{3\}}\nonumber\\& = 0,\\
\Lcalonethreet^{\{2,3\}} &= \mathrm{ln}\left( \mathrm{Tr}_{1} \left[e^{\Lambdaonehat }\left( \mathds{1}_{\{2,3\}}\otimes\frac{\mathds{1}_{\{1\}}}{2} \right)\right]\otimes \mathds{1}_{\{1\}} \right) \nonumber\\
&\hspace{0.5cm} - \Lcalonethreet^{\{2\}}   - \Lcalonethreet^{\{3\}}\nonumber\\& = 0,\\
\Lcalonethreet^{\{1,2,3\}} &= \Lcalonethreet - \Lcalonethreet^{\{1\}}  - \Lcalonethreet^{\{2\}} - \Lcalonethreet^{\{3\}} \nonumber\\
&\hspace{0.5cm}-\Lcalonethreet^{\{1,2\}}  - \Lcalonethreet^{\{2,3\}} - \Lcalonethreet^{\{1,3\}}\nonumber\\& = 0.
\end{align}
\end{subequations}
As expected, the cluster expansion gives only a non-trivial Lindbladian acting on qubit 1. 
In evaluating Eq.~\eqref{eq:exampleone}, we have used the fact that $\mathrm{Tr}\left[e^{\Lambdaonehat}\rho\right] = 1$, since $e^{\Lambdaonehat}$ is a CPTP map which preserves the trace of density matrices.

As a second example, let us consider the Lindbladian of a three-qubit quantum processor which comprises a jump operator acting on both qubits 1 and 2, i.e.,
\begin{equation}
\Lcalonethreet = \Lambdaonetwohat,
\end{equation}
where
\begin{equation}
\Lambdaonetwo (\rho) =\gamma_{1,2}\left(\Aonetwohat \rho \Aonetwohatd -\frac{1}{2}\Aonetwohatd\Aonetwohat\rho - \frac{1}{2}\rho\Aonetwohatd\Aonetwohat \right) 
\end{equation}
consists of a correlated jump operator $\Aonetwohat$ acting on both qubits 1 and 2. We decompose the Lindbladian using the cluster expansion [Eq.~\eqref{eq:recur}] as 
\begin{subequations}
\begin{align}
\Lcalonethreet^{\{1\}} &= \mathrm{ln}\left( \mathrm{Tr}_{23} \left[e^{\Lambdaonetwohat }\left( \mathds{1}_{\{1\}}\otimes\frac{\mathds{1}_{\{2,3\}}}{4} \right)\right]\otimes \mathds{1}_{\{2,3\}} \right)\nonumber\\
& = \mathrm{ln}\left(\mathrm{Tr}_{2}\left[  \frac{1}{2} e^{\Lambdaonetwohat} \right]\right)  \equiv \langle \Lambdaonetwohat \rangle_2, \\
\Lcalonethreet^{\{2\}} &= \mathrm{ln}\left( \mathrm{Tr}_{13} \left[e^{\Lambdaonetwohat }\left( \mathds{1}_{\{2\}}\otimes\frac{\mathds{1}_{\{1,3\}}}{4} \right)\right]\otimes \mathds{1}_{\{1,3\}} \right)\nonumber\\
& = \mathrm{ln}\left(\mathrm{Tr}_{1}\left[  \frac{1}{2} e^{\Lambdaonetwohat} \right]\right)  \equiv \langle \Lambdaonetwohat \rangle_1, \\
\Lcalonethreet^{\{1,2\}} &= \Lambdaonetwohat-  \langle\Lambdaonetwohat \rangle_2 -  \langle\Lambdaonetwohat \rangle_1,\\
\Lcalonethreet^{\{3\}} &= \Lcalonethreet^{\{2,3\}} = \Lcalonethreet^{\{1,3\}} = \Lcalonethreet^{\{1,2,3\}} = 0.
\end{align}
\end{subequations}
The total Lindbladian is therefore decomposed into components that act on qubit 1, qubit 2 and both qubits 1 and 2.  The single-qubit components $\Lcalonent^{S_1}$ (e.g., $\Lcalonethreet^{\{1\}}, \Lcalonethreet^{\{2\}}$) of the Lindbladian   can be thought of as parts of the Lindbladian that act on each qubit irrespective of the other qubits states, and the two-qubit components $\Lcalonent^{S_2}$ (e.g., $\Lcalonethreet^{\{1,2\}}$)  of the Lindbladian describe pure two-qubit correlations.  Note that while the single-qubit components $\Lcalonent^{S_{1}}$ of the Lindbladian  obtained from the cluster expansion approach are guaranteed to be a physical Lindbladian, the higher-order correlated components $\Lcalonent^{S_{m>1}}$ are not guaranteed to be physical. However, in this work we find that we can construct a CPTP approximate noise channel by choosing appropriate gain factors for each of the Lindbladian components in Eq.~\eqref{eq:Napproxm}.

\section{Algorithm to simulate $\dsl2,0,2\dsr$ code}\label{sec:algorithmtwo}
In Algorithm~\ref{algo_202} below, we provide the pseudocode for simulating the $\dsl2,0,2\dsr$ code using either the actual or approximate noise channels.
\begin{widetext}
\begin{algorithm}\small
\caption{$\dsl2,0,2\dsr$ code simulation}\label{algo_202}
\Input{Quantum channel $\mathcal{V}^{ZZ}_{\mathcal{E}}$ and $\mathcal{V}^{XX}_{\mathcal{E}}$ of the $ZZ$ and $XX$ parity checks for $\mathcal{E} = \mathrm{actual/approx}$, and an integer $l$ = number of QEC rounds}
\Output{Set of data-qubit density matrices $\{\rhoe(x|\Psiin)\}$ and syndrome-qubit measurement probability distributions $\{\pe(x|\Psiin)\}$}
\Fn{\FMaintwo{$\mathcal{V}^{ZZ}_{\mathcal{E}}$, $\mathcal{V}^{XX}_{\mathcal{E}}$, $l$}}{
\For{$\Psiin \in \{ \Phi^{+}, \Phi^{-}, \Psi^{-}, \Psi^{+}\}$}{
$x_0 = ``0"$ \\
$x_{\mathrm{set}} = \{x_0\}$\\
\tcp{Initialize data-qubits to one of the the Bell states}
Set $\rhoe (x_0|\Psiin) \gets |\Psiin\rangle\langle \Psiin|$\\
\For{$i \gets 1$ \KwTo  $l$}{
Reset syndrome qubit to $|0\rangle$\\
$x_{\mathrm{tempset}} = \emptyset $\\
\tcp{$ZZ$ check}
\For{$x \in x_{\mathrm{set}}$}{
\tcp{Subscript $a$ refers to the ancilla qubit}
$\rho_{\mathrm{sys}} \gets \mathcal{V}^{ZZ}_\mathcal{E} \big(\rhoe (x|\Psiin) \otimes |0\rangle\langle 0|_a\big)$\\
\For{$ m \in \{``0",``1"\}$}{
\tcp{Keep track of measurement results}
\If{$i >1$}{
$x \gets x + m$}
\Else{
$x \gets m$}
$x_{\mathrm{tempset}} = x_{\mathrm{tempset}} \cup \{x\} $\\
\tcp{Measure ancilla qubit with operator $M_a = |m\rangle\langle m|_a$ for $m = 0,1$}
$\rhoe(x|\Psiin) \gets \mathrm{Tr}_a\left[(\mathds{1}\otimes |m\rangle \langle m|_a) \rho_{\mathrm{sys}}\right]$\\
\tcp{Syndrome measurement probability}
$\pe(x|\Psiin) \gets \mathrm{Tr} \left[ \rhoe(x|\Psiin) \right]$ \\
Store $\pe(x|\Psiin)$\\
\tcp{Normalization of data-qubit density matrix}
$\rhoe(x|\Psiin) \gets  \rhoe(x|\Psiin)/\pe(x|\Psiin)$\\
Store $\rhoe(x|\Psiin)$\\
}}
$x_{\mathrm{set}} \gets x_{\mathrm{tempset}} $\\
$x_{\mathrm{tempset}} \gets \emptyset $\\
\tcp{$XX$ check}
\For{$x \in x_{\mathrm{set}}$}{
$x_{\mathrm{last}} \gets \Fgetlastchar(x)$\\
$\rho_{\mathrm{sys}} \gets \mathcal{V}^{XX}_\mathcal{E} \big(\rhoe (x|\Psiin) \otimes |x_{\mathrm{last}}\rangle\langle x_{\mathrm{last}}|_a\big)$\\
\For{$ m \in \{``0",``1"\}$}{
\tcp{Keep track of measurement results}
$x \gets x + m$ \\
$x_{\mathrm{tempset}} = x_{\mathrm{tempset}} \cup \{x\} $\\
\tcp{Measure ancilla qubit with operator $M_a = |m\rangle\langle m|_a$ for $m = 0,1$}
$\rhoe(x|\Psiin) \gets \mathrm{Tr}_{a}\left[(\mathds{1}\otimes |m\rangle \langle m|_a) \rho_{\mathrm{sys}}\right]$\\
\tcp{Syndrome measurement probability}
$\pe(x|\Psiin) \gets \mathrm{Tr} \left[ \rhoe(x|\Psiin) \right]$ \\
Store $\pe(x|\Psiin)$\\
\tcp{Normalization of data-qubit density matrix}
$\rhoe(x|\Psiin) \gets  \rhoe(x|\Psiin)/\pe(x|\Psiin) $\\
Store $\rhoe(x|\Psiin)$
}}
$x_{\mathrm{set}} \gets x_{\mathrm{tempset}} $\\
}
}
\Return $\rhoe (x|\Psiin)$ and $\pe(x|\Psiin)$ for all $x$ and $\Psiin$}
\end{algorithm}
\end{widetext}

\section{Pulse sequences}\label{sec:pulse_sequence}
In this Appendix, we elaborate the details of the pulse sequences for the $ZZ$ and $XX$ parity checks as depicted in Fig.~\ref{fig:pulse}.  The pulse sequences consist of single-qubit and two-qubit (CNOT) gate pulses [see the labels on the top of panels (a) and (d) of Fig.~\ref{fig:pulse}]. In particular, our two-qubit CNOT gate is implemented using the cross-resonance (CR)~\cite{Rigetti2010Fully,Chow2011Simple} scheme, which is by driving the control qubit at the resonant frequency of the target qubit. Since the CR gate has a high gate speed and low static $ZZ$ coherent error when the frequency of the control qubit is higher~\cite{malekakhlagh2020first} than that of its corresponding target qubit, we choose the qubit frequency arrangement shown in Table~\ref{tab:params} and implement the equivalent parity check circuits shown in the right-hand sides of Fig.~\ref{fig:ZZ_XX_parity_circuit}.

\begin{figure}[t]
\includegraphics[width=\linewidth]{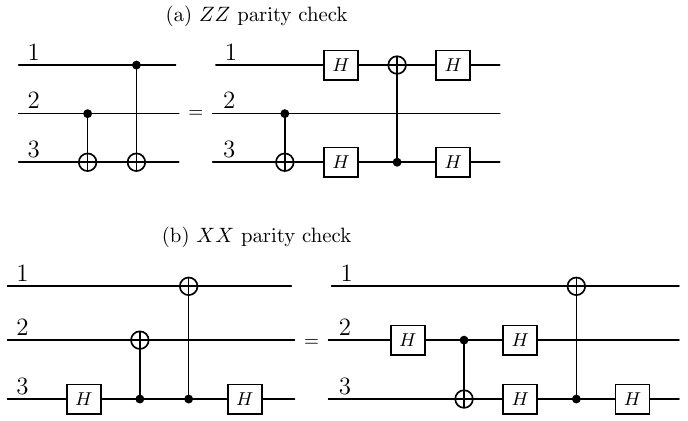}
	\caption{Circuits for implementing (a) $ZZ$ and (b) $XX$ parity checks. To achieve high gate speed and low coherent error of the cross-resonance gates, we reverse the directions of some of the CNOT gates, as shown by the equivalent circuits on the right-hand sides.}\label{fig:ZZ_XX_parity_circuit}
\end{figure} 

An ideal CR gate is equivalent to the $R_{ZX}(\pi/2)$ rotation, which is a conditional rotation of the target qubit around the $X$ axis, where the rotation direction (clockwise/counterclockwise) depends on the state of the control qubit. The action of $R_{ZX}(\pi/2)$ gate can therefore be written as
\begin{equation}
R_{ZX}(\pi/2) \equiv \exp\left(-i\frac{\pi}{4} Z \otimes X\right) = \frac{1}{\sqrt{2}} \left(
\begin{matrix}
1&-i&0&0\\
-i&1&0&0\\
0&0&1&i\\
0&0&i&1
\end{matrix}
\right),
\end{equation}
where the basis ordering of the matrix above is $|00\rangle$, $|01\rangle$, $|10\rangle$, and $|11\rangle$. A CNOT gate is equivalent to the $R_{ZX}(\pi/2)$ gate up to single-qubit rotations and a global phase factor, i.e.,
\begin{equation}\label{eq:CNOT}
\mathrm{CNOT}_{12} = e^{-i\pi/4}  R_{ZX}(\pi/2)\cdot \left[R_Z(-\pi/2)\otimes R_X(-\pi/2)\right],
\end{equation}
where
\begin{equation}
\mathrm{CNOT}_{12} = \left(
\begin{matrix}
1&0&0&0\\
0&1&0&0\\
0&0&0&1\\
0&0&1&0
\end{matrix}
\right).
\end{equation}
In the above, we have used the notation $\mathrm{CNOT}_{ab}$ to indicate a CNOT gate with the control qubit being qubit $a$ and target qubit being qubit $b$. 

\begin{figure}[t]
\includegraphics[width=\linewidth]{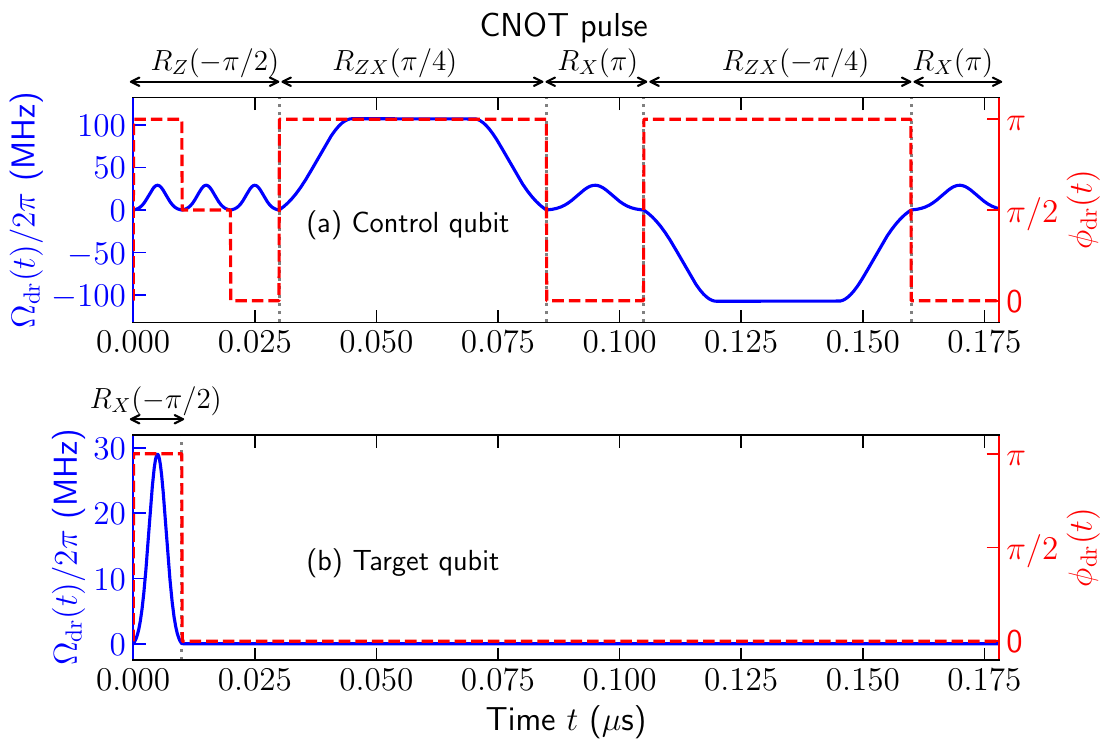}
	\caption{Pulse sequences of the CNOT gate for (a) control and (b) target qubit. The pulse sequence comprises single-qubit rotations and the echoed version of the CR pulse as given in Eqs.~\eqref{eq:CNOT} and~\eqref{eq:CRpulse}.}\label{fig:CNOT_pulse}
\end{figure} 

As in Ref.~\cite{Sheldon2016Procedure}, we implement an echoed version of the CR gate where the $R_{ZX}(\pi/2)$ gate is decomposed into two opposite-sign half-rotations, i.e., $R_{ZX}(\pi/4)$ and $R_{ZX}(-\pi/4)$ rotations, and single-qubit gates $(X\otimes I)$ in the following way:
\begin{equation}\label{eq:CRpulse}
R_{ZX}(\pi/2) = (X\otimes I) \cdot R_{ZX} (-\pi/4) \cdot (X\otimes I) \cdot R_{ZX} (\pi/4). 
\end{equation}
This decomposition makes the echoed version of the CR gate more resilient to coherent errors common to the $R_{ZX} (\pi/4)$ and $ R_{ZX} (-\pi/4)$ rotations where it partially cancels the unwanted $ZZ$ interaction between the qubits. The CNOT gate obtained using this echoed CR pulse is shown in Fig.~\ref{fig:CNOT_pulse}. Here, we implement the single-qubit $R_{Z} (-\pi/2)$ gates  in the CNOT pulse sequence of Fig.~\ref{fig:CNOT_pulse} and the Hadamard gate in the parity-check pulse sequences (Fig.~\ref{fig:pulse}) using sequences of $X$ and $Y$ rotation pulses obtained via the $X$-$Y$ decomposition of the gates~\cite{nielsen2010quantum}. The pulse sequences for the $R_{Z} (-\pi/2)$ and Hadamard gate are depicted in Fig.~\ref{fig:Z_H_pulse}.

In our simulation, we choose the $X$ and $Y$ rotation pulses to have Gaussian envelopes given by
\begin{align}\label{eq:Gaussian}
\Omega_{\mathrm{dr}}(t) &= g(t;A,t_0,\tgate,\sigma) \nonumber\\
&=A \exp\left(- \frac{(t - (t_0 + \tgate/2))^2}{2\sigma^2} \right),
\end{align}
where $A$, $t_0$, and $\tgate$ are amplitude, starting time and duration of the pulse, respectively. Here, we set the width of the Gaussian pulse $\sigma$ to be $\tgate/8$. Additionally, we implement the $R_{ZX}(-\pi/4)$ and $R_{ZX}(\pi/4)$ of the two-qubit gates using the Gaussian-square pulse with a rise and fall time of $3\sigma = 15$ ns, where it can be written using a piecewise function:
\begin{widetext}
\begin{align}
\Omega_{\mathrm{dr}}(t) = \begin{cases}
g(t;A,t_0,6\sigma,\sigma), & \mathrm{for \quad} t_0 \leq t<t_0+3\sigma,\\ 
A, & \mathrm{for \quad} t_0+3\sigma \leq t \leq t_0+\tgate -3\sigma, \\
g(t;A,t_0+\tgate -6\sigma,6\sigma,\sigma), & \mathrm{for \quad}  t_0+\tgate -3\sigma <t \leq t_0+\tgate,
\end{cases}
\end{align}
\end{widetext}
with $g(t;A,t_0,\tgate,\sigma)$ being the Gaussian function as defined in Eq.~\eqref{eq:Gaussian}. 

Each of the pulses is parameterized by a pulse amplitude, duration and frequency. The pulses shown in Fig.~\ref{fig:pulse} are the pulses which are obtained by optimizing the fidelity of each cross resonance and single-qubit gate where the state-averaged fidelity is given by~\cite{Wood2018quantification,Rol2019Fast} 
\begin{align}
F &= \int_{\psi \in \chi} \langle \psi |U^\dagger \Vactual\left( |\psi\rangle \langle\psi |\right) U|\psi\rangle \nonumber\\
&=  \frac{\dchi + \sum_k|\mathrm{Tr}(U^\dagger E_k)|^2}{\dchi(\dchi+1)}.
\end{align}
Here, $\dchi$ is the dimension of Hilbert space $\chi$ of all the qubits involved in the parity checks, $U$ is the target unitary of the ideal gate and $\{E_k\}$ are the Kraus operators corresponding to the superoperator $\Vactual$ of the actual implementation of the gate. This fidelity can be calculated by using the ``average\_gate\_fidelity" function in QuTiP~\cite{Johansson2012Qutip,Johansson2013Qutip}. Each of the pulses in the optimal pulse sequence for both parity checks is obtained by first fixing the amplitude and duration of the pulse and then scanning over the pulse frequency to find the frequency that maximizes the target gate fidelity. Using this optimal frequency, we then scan over the amplitude of the pulse to obtain the amplitude for each of the pulses that maximizes the gate fidelity.

\begin{figure}[t]
\includegraphics[width=\linewidth]{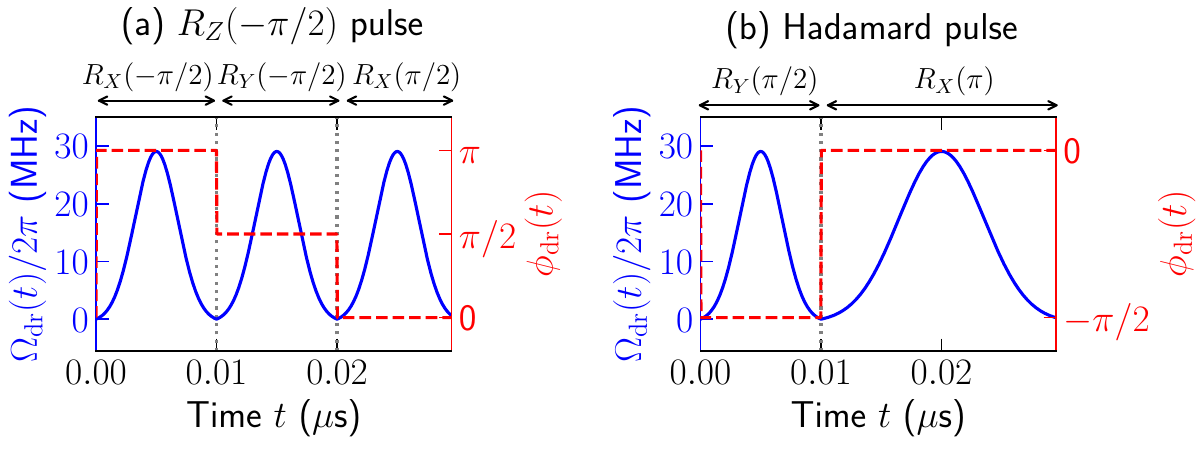}
	\caption{Pulse sequences for the (a) $R_Z(-\pi/2)$ and (b) H gate in terms of the elementary $R_X(\theta)$ and $R_Y(\theta)$ rotation pulses.}\label{fig:Z_H_pulse}
\end{figure}

\section{Simulation results for the triangle geometry}\label{sec:triangle}
In this Appendix, we present results calculated for the 8th round of the $\dsl2,0,2\dsr$ code in the triangle geometry. Figure~\ref{fig:triangle_8round} shows the infidelities between different quantum channels where the approximate channels are simulated using the second- and third-order approximations. The overall qualitative behaviors of the results for the triangle geometry are similar to those of the linear geometry shown in Fig.~\ref{fig:linear_8round} of Sec.~\ref{sec:results}.

\begin{figure*}[t!]
\includegraphics[width=\linewidth]{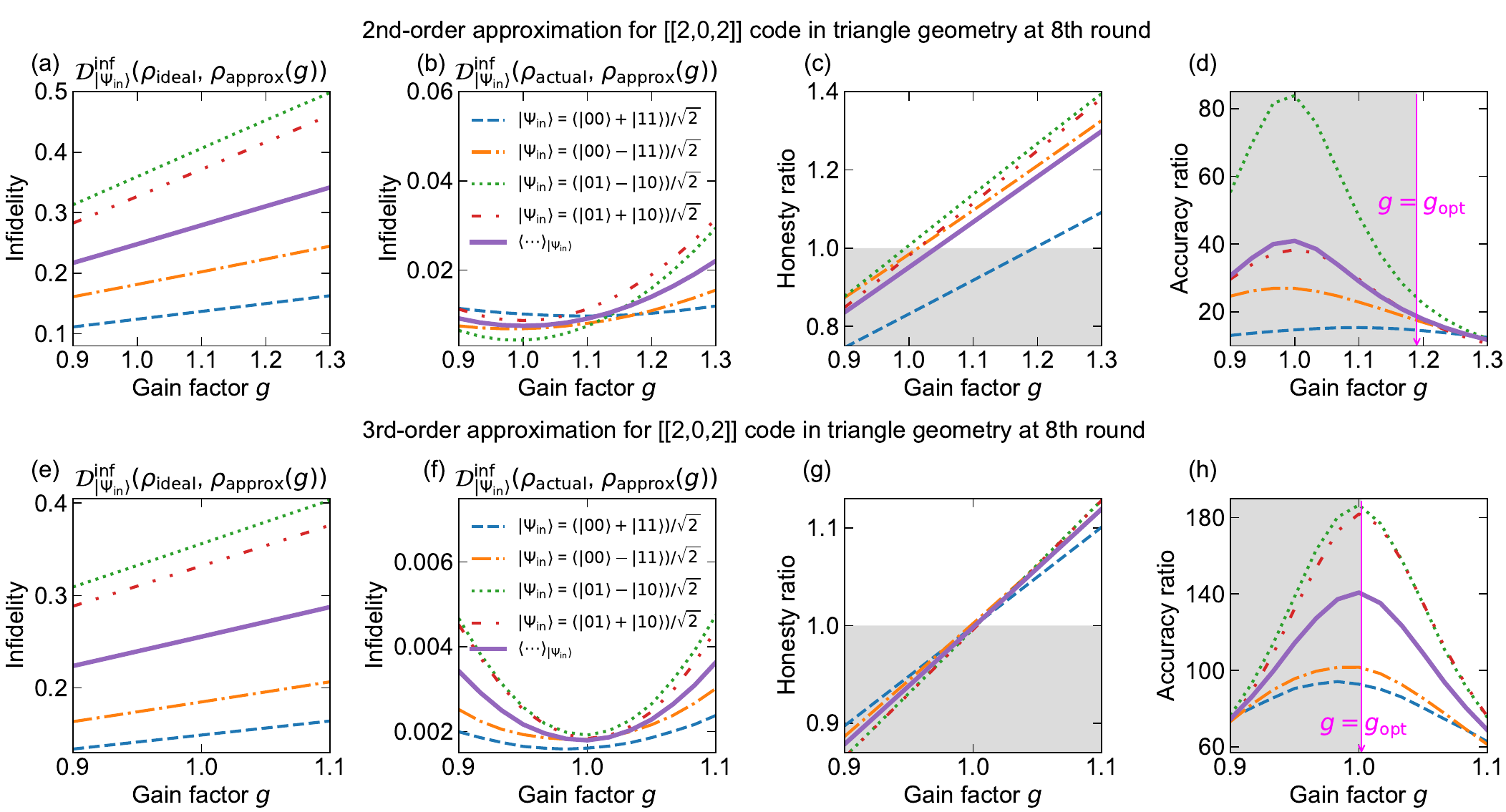}
	\caption{Results simulated for the triangular connectivity using the second-order $\Napproxhattwo$ [upper panels: (a), (b), (c) and (d)] and third-order $\Napproxhatthree$ [lower panels: (e), (f), (g) and (h)] approximate noise channels. They are calculated at the 8th round implementation of the $\dsl2,0,2\dsr$ code for different input data-qubits states $|\Psiin\rangle$; the state-averaged results are denoted by $\langle \cdots\rangle_{|\Psiin\rangle}$ and shown as solid purple lines. (a,b,e,f) Infidelity between the data-qubits density matrices as a function of gain factor $g$: (a,e) Infidelity between ideal and approximate density matrices, (b,f) Infidelity between actual and approximate density matrices. (c,g) Honesty ratio as a function of gain factor $g$. (d,h) Accuracy ratio as a function of gain factor $g$. Magenta arrow denotes $g = \gopt$, i.e., the optimal value of $g$ that maximizes the state-averaged accuracy ratio subject to the honesty criterion [Eq.~\eqref{eq:honesty}], where $\gopt \simeq 1.2$ for the second-order and $\gopt \simeq 1.002$ for the third-order approximate channel. Shaded areas in panels (c), (d), (g), and (h) denote regimes where approximate noise models are not honest (honesty ratio $< 1$).}\label{fig:triangle_8round}
\end{figure*} 

\end{appendix}
\bibliographystyle{quantum}
\bibliography{reference}

\begin{thebibliography}{10}

\bibitem{Shor1995Scheme}
Peter~W. Shor.
\newblock ``Scheme for reducing decoherence in quantum computer memory''.
\newblock \href{https://dx.doi.org/10.1103/PhysRevA.52.R2493}{Phys. Rev. A {\bf
  52}, R2493--R2496}~(1995).

\bibitem{shor1996fault}
Peter~W Shor.
\newblock ``Fault-tolerant quantum computation''.
\newblock In Proceedings of 37th conference on foundations of computer science.
\newblock \href{https://dx.doi.org/10.1109/SFCS.1996.548464}{Pages 56--65}.
\newblock IEEE~(1996).

\bibitem{Calderbank1996Good}
A.~R. Calderbank and Peter~W. Shor.
\newblock ``Good quantum error-correcting codes exist''.
\newblock \href{https://dx.doi.org/10.1103/PhysRevA.54.1098}{Phys. Rev. A {\bf
  54}, 1098}~(1996).

\bibitem{devitt2013quantum}
Simon~J. Devitt, William~J. Munro, and Kae Nemoto.
\newblock ``Quantum error correction for beginners''.
\newblock \href{https://dx.doi.org/10.1088/0034-4885/76/7/076001}{Reports on
  Progress in Physics {\bf 76}, 076001}~(2013).

\bibitem{aharonov1997fault}
Dorit Aharonov and Michael Ben-Or.
\newblock ``Fault-tolerant quantum computation with constant error''.
\newblock In Proceedings of the twenty-ninth annual ACM symposium on Theory of
  computing.
\newblock \href{https://dx.doi.org/10.1137/S0097539799359385}{Pages 176--188}.
\newblock ~(1997).

\bibitem{Gottesman1998Theory}
Daniel Gottesman.
\newblock ``Theory of fault-tolerant quantum computation''.
\newblock \href{https://dx.doi.org/10.1103/PhysRevA.57.127}{Phys. Rev. A {\bf
  57}, 127--137}~(1998).

\bibitem{preskill1998reliable}
John Preskill.
\newblock ``Reliable quantum computers''.
\newblock \href{https://dx.doi.org/10.1098/rspa.1998.0167}{Proceedings of the
  Royal Society of London. Series A: Mathematical, Physical and Engineering
  Sciences {\bf 454}, 385--410}~(1998).

\bibitem{knill1998resilient}
Emanuel Knill, Raymond Laflamme, and Wojciech~H Zurek.
\newblock ``Resilient quantum computation''.
\newblock \href{https://dx.doi.org/10.1126/science.279.5349.342}{Science {\bf
  279}, 342}~(1998).

\bibitem{Guillard2019Repetition}
J\'er\'emie Guillaud and Mazyar Mirrahimi.
\newblock ``Repetition cat qubits for fault-tolerant quantum computation''.
\newblock \href{https://dx.doi.org/10.1103/PhysRevX.9.041053}{Phys. Rev. X {\bf
  9}, 041053}~(2019).

\bibitem{Darmawan2021Practical}
Andrew~S. Darmawan, Benjamin~J. Brown, Arne~L. Grimsmo, David~K. Tuckett, and
  Shruti Puri.
\newblock ``Practical quantum error correction with the xzzx code and kerr-cat
  qubits''.
\newblock \href{https://dx.doi.org/10.1103/PRXQuantum.2.030345}{PRX Quantum
  {\bf 2}, 030345}~(2021).

\bibitem{Chamberland2022Building}
Christopher Chamberland, Kyungjoo Noh, Patricio Arrangoiz-Arriola, Earl~T.
  Campbell, Connor~T. Hann, Joseph Iverson, Harald Putterman, Thomas~C.
  Bohdanowicz, Steven~T. Flammia, Andrew Keller, Gil Refael, John Preskill,
  Liang Jiang, Amir~H. Safavi-Naeini, Oskar Painter, and Fernando~G.S.L.
  Brand\~ao.
\newblock ``Building a fault-tolerant quantum computer using concatenated cat
  codes''.
\newblock \href{https://dx.doi.org/10.1103/PRXQuantum.3.010329}{PRX Quantum
  {\bf 3}, 010329}~(2022).

\bibitem{bonilla2021xzzx}
J.~Pablo Bonilla~Ataides, David~K. Tuckett, Stephen~D. Bartlett, Steven~T.
  Flammia, and Benjamin~J. Brown.
\newblock ``The xzzx surface code''.
\newblock \href{https://dx.doi.org/10.1038/s41467-021-22274-1}{Nature
  communications {\bf 12}, 2172}~(2021).

\bibitem{blume2017demonstration}
Robin Blume-Kohout, John~King Gamble, Erik Nielsen, Kenneth Rudinger, Jonathan
  Mizrahi, Kevin Fortier, and Peter Maunz.
\newblock ``Demonstration of qubit operations below a rigorous fault tolerance
  threshold with gate set tomography''.
\newblock \href{https://dx.doi.org/10.1038/ncomms14485}{Nature communications
  {\bf 8}, 14485}~(2017).

\bibitem{nielsen2021gate}
Erik Nielsen, John~King Gamble, Kenneth Rudinger, Travis Scholten, Kevin Young,
  and Robin Blume-Kohout.
\newblock ``Gate set tomography''.
\newblock \href{https://dx.doi.org/10.22331/q-2021-10-05-557}{Quantum {\bf 5},
  557}~(2021).

\bibitem{proctor2020detecting}
Timothy Proctor, Melissa Revelle, Erik Nielsen, Kenneth Rudinger, Daniel
  Lobser, Peter Maunz, Robin Blume-Kohout, and Kevin Young.
\newblock ``Detecting and tracking drift in quantum information processors''.
\newblock \href{https://dx.doi.org/10.1038/s41467-020-19074-4}{Nature
  communications {\bf 11}, 5396}~(2020).

\bibitem{paris2004quantum}
Matteo Paris and Jaroslav Rehacek.
\newblock ``Quantum state estimation''.
\newblock \href{https://dx.doi.org/https://doi.org/10.1007/b98673}{Springer
  Science \& Business Media}. Heidelberg, Germany~(2004).

\bibitem{howard2006quantum}
M.~Howard, J.~Twamley, C.~Wittmann, T.~Gaebel, F.~Jelezko, and J.~Wrachtrup.
\newblock ``Quantum process tomography and linblad estimation of a solid-state
  qubit''.
\newblock \href{https://dx.doi.org/10.1088/1367-2630/8/3/033}{New Journal of
  Physics {\bf 8}, 33}~(2006).

\bibitem{Samach2022Lindblad}
Gabriel~O. Samach, Ami Greene, Johannes Borregaard, Matthias Christandl, Joseph
  Barreto, David~K. Kim, Christopher~M. McNally, Alexander Melville, Bethany~M.
  Niedzielski, Youngkyu Sung, Danna Rosenberg, Mollie~E. Schwartz, Jonilyn~L.
  Yoder, Terry~P. Orlando, Joel I-Jan Wang, Simon Gustavsson, Morten
  Kjaergaard, and William~D. Oliver.
\newblock ``Lindblad tomography of a superconducting quantum processor''.
\newblock \href{https://dx.doi.org/10.1103/PhysRevApplied.18.064056}{Phys. Rev.
  Appl. {\bf 18}, 064056}~(2022).

\bibitem{emerson2005scalable}
Joseph Emerson, Robert Alicki, and Karol {\.Z}yczkowski.
\newblock ``Scalable noise estimation with random unitary operators''.
\newblock \href{https://dx.doi.org/10.1088/1464-4266/7/10/021}{Journal of
  Optics B: Quantum and Semiclassical Optics {\bf 7}, S347--S352}~(2005).

\bibitem{Magesan2011Scalable}
Easwar Magesan, J.~M. Gambetta, and Joseph Emerson.
\newblock ``Scalable and robust randomized benchmarking of quantum processes''.
\newblock \href{https://dx.doi.org/10.1103/PhysRevLett.106.180504}{Phys. Rev.
  Lett. {\bf 106}, 180504}~(2011).

\bibitem{harper2020efficient}
Robin Harper, Steven~T. Flammia, and Joel~J. Wallman.
\newblock ``Efficient learning of quantum noise''.
\newblock \href{https://dx.doi.org/10.1038/s41567-020-0992-8}{Nature Physics
  {\bf 16}, 1184}~(2020).

\bibitem{van2023probabilistic}
Ewout Van Den~Berg, Zlatko~K. Minev, Abhinav Kandala, and Kristan Temme.
\newblock ``Probabilistic error cancellation with sparse pauli--lindblad models
  on noisy quantum processors''.
\newblock \href{https://dx.doi.org/10.1038/s41567-023-02042-2}{Nature Physics
  {\bf 19}, 1116}~(2023).

\bibitem{Harper2023Learning}
Robin Harper and Steven~T. Flammia.
\newblock ``Learning correlated noise in a 39-qubit quantum processor''.
\newblock \href{https://dx.doi.org/10.1103/PRXQuantum.4.040311}{PRX Quantum
  {\bf 4}, 040311}~(2023).

\bibitem{emerson2007symmetrized}
Joseph Emerson, Marcus Silva, Osama Moussa, Colm Ryan, Martin Laforest,
  Jonathan Baugh, David~G Cory, and Raymond Laflamme.
\newblock ``Symmetrized characterization of noisy quantum processes''.
\newblock \href{https://dx.doi.org/10.1126/science.1145699}{Science {\bf 317},
  1893}~(2007).

\bibitem{Silva2008Scalable}
M.~Silva, E.~Magesan, D.~W. Kribs, and J.~Emerson.
\newblock ``Scalable protocol for identification of correctable codes''.
\newblock \href{https://dx.doi.org/10.1103/PhysRevA.78.012347}{Phys. Rev. A
  {\bf 78}, 012347}~(2008).

\bibitem{onorati2021fitting}
Emilio Onorati, Tamara Kohler, and Toby Cubitt.
\newblock ``Fitting quantum noise models to tomography data''.
\newblock \href{https://dx.doi.org/10.22331/q-2023-12-05-1197}{Quantum {\bf 7},
  1197}~(2023).

\bibitem{onorati2023fitting}
{Onorati, Emilio and Kohler, Tamara and Cubitt, Toby S}.
\newblock ``{Fitting time-dependent Markovian dynamics to noisy quantum
  channels}''.
\newblock
  \href{https://dx.doi.org/https://doi.org/10.48550/arXiv.2303.08936}{{arXiv:2303.08936}}~(2023).

\bibitem{lindblad1976generators}
Goran Lindblad.
\newblock ``On the generators of quantum dynamical semigroups''.
\newblock \href{https://dx.doi.org/10.1007/BF01608499}{Communications in
  Mathematical Physics {\bf 48}, 119--130}~(1976).

\bibitem{gorini1976completely}
Vittorio Gorini, Andrzej Kossakowski, and Ennackal Chandy~George Sudarshan.
\newblock ``Completely positive dynamical semigroups of n-level systems''.
\newblock \href{https://dx.doi.org/10.1063/1.522979}{Journal of Mathematical
  Physics {\bf 17}, 821--825}~(1976).

\bibitem{groszkowski2023simple}
Peter Groszkowski, Alireza Seif, Jens Koch, and Aashish~A. Clerk.
\newblock ``Simple master equations for describing driven systems subject to
  classical non-markovian noise''.
\newblock \href{https://dx.doi.org/10.22331/q-2023-04-06-972}{Quantum {\bf 7},
  972}~(2023).

\bibitem{fleming2012non}
C.~H. Fleming and B.~L. Hu.
\newblock ``Non-markovian dynamics of open quantum systems: stochastic
  equations and their perturbative solutions''.
\newblock \href{https://dx.doi.org/10.1016/j.aop.2011.12.006}{Annals of Physics
  {\bf 327}, 1238--1276}~(2012).

\bibitem{Gulacsi2023signatures}
Bal\'azs Gul\'acsi and Guido Burkard.
\newblock ``Signatures of non-markovianity of a superconducting qubit''.
\newblock \href{https://dx.doi.org/10.1103/PhysRevB.107.174511}{Phys. Rev. B
  {\bf 107}, 174511}~(2023).

\bibitem{navarrete2015open}
{Navarrete-Benlloch, Carlos}.
\newblock ``{Open systems dynamics: Simulating master equations in the
  computer}''.
\newblock
  \href{https://dx.doi.org/https://doi.org/10.48550/arXiv.1504.05266}{{arXiv:1504.05266}}~(2015).

\bibitem{mayer1941molecular}
Joseph~E Mayer and Elliott Montroll.
\newblock ``Molecular distribution''.
\newblock \href{https://dx.doi.org/10.1063/1.1750822}{The Journal of Chemical
  Physics {\bf 9}, 2--16}~(1941).

\bibitem{kira2011semiconductor}
Mackillo Kira and Stephan~W Koch.
\newblock ``Semiconductor quantum optics''.
\newblock
  \href{https://dx.doi.org/https://doi.org/10.1017/CBO9781139016926}{Cambridge
  University Press}. Cambridge, UK~(2011).

\bibitem{Magesan2013Modeling}
Easwar Magesan, Daniel Puzzuoli, Christopher~E. Granade, and David~G. Cory.
\newblock ``Modeling quantum noise for efficient testing of fault-tolerant
  circuits''.
\newblock \href{https://dx.doi.org/10.1103/PhysRevA.87.012324}{Phys. Rev. A
  {\bf 87}, 012324}~(2013).

\bibitem{Puzzuoli2014Tractable}
Daniel Puzzuoli, Christopher Granade, Holger Haas, Ben Criger, Easwar Magesan,
  and D.~G. Cory.
\newblock ``Tractable simulation of error correction with honest approximations
  to realistic fault models''.
\newblock \href{https://dx.doi.org/10.1103/PhysRevA.89.022306}{Phys. Rev. A
  {\bf 89}, 022306}~(2014).

\bibitem{Mauricio2015Comparison}
Mauricio Guti\'errez and Kenneth~R. Brown.
\newblock ``Comparison of a quantum error-correction threshold for exact and
  approximate errors''.
\newblock \href{https://dx.doi.org/10.1103/PhysRevA.91.022335}{Phys. Rev. A
  {\bf 91}, 022335}~(2015).

\bibitem{corcoles2015demonstration}
Antonio~D. C{\'o}rcoles, Easwar Magesan, Srikanth~J. Srinivasan, Andrew~W.
  Cross, Matthias Steffen, Jay~M. Gambetta, and Jerry~M. Chow.
\newblock ``Demonstration of a quantum error detection code using a square
  lattice of four superconducting qubits''.
\newblock \href{https://dx.doi.org/10.1038/ncomms7979}{Nature communications
  {\bf 6}, 6979}~(2015).

\bibitem{jozsa1994fidelity}
Richard Jozsa.
\newblock ``Fidelity for mixed quantum states''.
\newblock \href{https://dx.doi.org/10.1080/09500349414552171}{Journal of modern
  optics {\bf 41}, 2315--2323}~(1994).

\bibitem{nielsen2010quantum}
M.~A. Nielsen and I.~L. Chuang.
\newblock ``Quantum computation and quantum information''.
\newblock
  \href{https://dx.doi.org/https://doi.org/10.1017/CBO9780511976667}{Cambridge
  University Press}. Cambridge, UK~(2000).

\bibitem{Rigetti2010Fully}
Chad Rigetti and Michel Devoret.
\newblock ``Fully microwave-tunable universal gates in superconducting qubits
  with linear couplings and fixed transition frequencies''.
\newblock \href{https://dx.doi.org/10.1103/PhysRevB.81.134507}{Phys. Rev. B
  {\bf 81}, 134507}~(2010).

\bibitem{Chow2011Simple}
Jerry~M. Chow, A.~D. C\'orcoles, Jay~M. Gambetta, Chad Rigetti, B.~R. Johnson,
  John~A. Smolin, J.~R. Rozen, George~A. Keefe, Mary~B. Rothwell, Mark~B.
  Ketchen, and M.~Steffen.
\newblock ``Simple all-microwave entangling gate for fixed-frequency
  superconducting qubits''.
\newblock \href{https://dx.doi.org/10.1103/PhysRevLett.107.080502}{Phys. Rev.
  Lett. {\bf 107}, 080502}~(2011).

\bibitem{Cruz2021Testing}
Pedro M.~Q. Cruz and J.~Fern\'andez-Rossier.
\newblock ``Testing complementarity on a transmon quantum processor''.
\newblock \href{https://dx.doi.org/10.1103/PhysRevA.104.032223}{Phys. Rev. A
  {\bf 104}, 032223}~(2021).

\bibitem{cruz2023shallow}
Pedro M.~Q.~Cruz and Bruno Murta.
\newblock ``Shallow unitary decompositions of quantum fredkin and toffoli gates
  for connectivity-aware equivalent circuit averaging''.
\newblock \href{https://dx.doi.org/10.1063/5.0187026}{APL Quantum {\bf 1},
  016105}~(2024).

\bibitem{Sheldon2016Procedure}
Sarah Sheldon, Easwar Magesan, Jerry~M. Chow, and Jay~M. Gambetta.
\newblock ``Procedure for systematically tuning up cross-talk in the
  cross-resonance gate''.
\newblock \href{https://dx.doi.org/10.1103/PhysRevA.93.060302}{Phys. Rev. A
  {\bf 93}, 060302}~(2016).

\bibitem{Johansson2012Qutip}
J.~Robert Johansson, Paul~D. Nation, and Franco Nori.
\newblock ``Qutip: An open-source python framework for the dynamics of open
  quantum systems''.
\newblock \href{https://dx.doi.org/10.1016/j.cpc.2012.02.021}{Computer Physics
  Communications {\bf 183}, 1760}~(2012).

\bibitem{Johansson2013Qutip}
J.~Robert Johansson, Paul~D. Nation, and Franco Nori.
\newblock ``Qutip 2: A python framework for the dynamics of open quantum
  systems''.
\newblock
  \href{https://dx.doi.org/https://doi.org/10.1016/j.cpc.2012.11.019}{Computer
  Physics Communications {\bf 184}, 1234}~(2013).

\bibitem{magnus1954exponential}
Wilhelm Magnus.
\newblock ``On the exponential solution of differential equations for a linear
  operator''.
\newblock \href{https://dx.doi.org/10.1002/cpa.3160070404}{Communications on
  pure and applied mathematics {\bf 7}, 649--673}~(1954).

\bibitem{blanes2009magnus}
Sergio Blanes, Fernando Casas, Jose-Angel Oteo, and Jos{\'e} Ros.
\newblock ``The magnus expansion and some of its applications''.
\newblock \href{https://dx.doi.org/10.1016/j.physrep.2008.11.001}{Physics
  reports {\bf 470}, 151--238}~(2009).

\bibitem{krovi2023improved}
Hari Krovi.
\newblock ``Improved quantum algorithms for linear and nonlinear differential
  equations''.
\newblock \href{https://dx.doi.org/10.22331/q-2023-02-02-913}{Quantum {\bf 7},
  913}~(2023).

\bibitem{malekakhlagh2020first}
Moein Malekakhlagh, Easwar Magesan, and David~C. McKay.
\newblock ``First-principles analysis of cross-resonance gate operation''.
\newblock \href{https://dx.doi.org/10.1103/PhysRevA.102.042605}{Phys. Rev. A
  {\bf 102}, 042605}~(2020).

\bibitem{Wood2018quantification}
Christopher~J. Wood and Jay~M. Gambetta.
\newblock ``Quantification and characterization of leakage errors''.
\newblock \href{https://dx.doi.org/10.1103/PhysRevA.97.032306}{Phys. Rev. A
  {\bf 97}, 032306}~(2018).

\bibitem{Rol2019Fast}
M.~A. Rol, F.~Battistel, F.~K. Malinowski, C.~C. Bultink, B.~M. Tarasinski,
  R.~Vollmer, N.~Haider, N.~Muthusubramanian, A.~Bruno, B.~M. Terhal, and
  L.~DiCarlo.
\newblock ``Fast, high-fidelity conditional-phase gate exploiting leakage
  interference in weakly anharmonic superconducting qubits''.
\newblock \href{https://dx.doi.org/10.1103/PhysRevLett.123.120502}{Phys. Rev.
  Lett. {\bf 123}, 120502}~(2019).

\end{thebibliography}
\end{document}